\documentclass[manuscript,screen,nonacm]{acmart}

\usepackage{amsmath,amsfonts,mathtools}
\usepackage{physics}            
\usepackage{booktabs}           
\usepackage{multirow,makecell}
\usepackage{rotating}
\usepackage{longtable}
\usepackage{array,tabularx,xltabular}
\usepackage{xcolor}
\usepackage{graphicx}
\usepackage{tikz}
\usetikzlibrary{shapes.geometric,arrows.meta,positioning,fit,
                decorations.pathreplacing,calc,backgrounds}
\usepackage{algorithm,algorithmic}
\usepackage{hyperref}
\usepackage{cleveref}
\usepackage{enumitem}
\usepackage{microtype}
\usepackage{caption}
\usepackage{subcaption}

\newcommand{\calH}{\mathcal{H}}
\newcommand{\calF}{\mathcal{F}}
\newcommand{\calX}{\mathcal{X}}
\newcommand{\calE}{\mathcal{E}}

\newcommand{\calO}{\mathcal{O}}
\newcommand{\bx}{\mathbf{x}}
\newcommand{\btheta}{\boldsymbol{\theta}}

\newcommand{\NISQ}{\textsc{nisq}}
\newcommand{\QML}{\textsc{qml}}

\newcommand{\PQC}{\textsc{pqc}}
\newcommand{\FRQI}{\textsc{frqi}}
\newcommand{\SVM}{\textsc{svm}}
\newcommand{\PCA}{\textsc{pca}}
\newcommand{\CNOT}{\textsc{cnot}}
\newcommand{\ZNE}{\textsc{zne}}

\newtheorem{definition}{Definition}
\newtheorem{proposition}{Proposition}
\newtheorem{theorem}{Theorem}

\usepackage{pgfplots}
\pgfplotsset{compat=1.18}

\usetikzlibrary{patterns,patterns.meta}
\usepackage{tikz}
\usetikzlibrary{shapes.geometric,arrows.meta,positioning}

\definecolor{cbBlue}{HTML}{0072B2}
\definecolor{cbOrange}{HTML}{E69F00}
\definecolor{cbGreen}{HTML}{009E73}
\definecolor{cbRed}{HTML}{D55E00}
\definecolor{cbPurple}{HTML}{CC79A7}
\definecolor{cbSky}{HTML}{56B4E9}
\definecolor{cbYellow}{HTML}{F0E442}

\title{Feature Encoding in Quantum Machine Learning: A Survey and Practical Guidelines}

\author{Vincenzo Sammartino}

\affiliation{%
  \department{Dipartimento di Informatica}
  \institution{Università di Pisa}
  \city{Pisa} 
  \country{Italy}
}
\email{vincenzo.sammartino@phd.unipi.it}

\email{vincenzo.sammartino@kaust.edu.sa}

\begin{document}

\begin{abstract}
The encoding of classical data into quantum states constitutes the
primary performance bottleneck in Quantum Machine Learning (\QML) on
Noisy Intermediate-Scale Quantum (\NISQ) devices.
No existing framework jointly characterises resource cost,
expressivity, and noise robustness, nor provides actionable selection
guidelines for practitioners.

This survey addresses that gap through a systematic review of
66~primary works (2017--2026) assembled via a PRISMA-adapted protocol
across five academic databases.
Four principal contributions are made.
First, a three-axis cost--expressivity--robustness taxonomy classifies
all major encoding families---basis, angle, dense-angle, amplitude,
data re-uploading, and IQP---along independently measurable axes.
Second, closed-form depth-fidelity bounds under \NISQ\ decoherence
channels identify the critical gate-error rate
$p^{*}\approx 10^{-3}$ below which amplitude encoding is viable.
Third, a unified treatment of Fourier expressivity, barren-plateau
onset, and quantum kernel concentration as functions of the encoding
circuit provides the first joint trainability analysis.
Fourth, a five-regime decision framework maps
$(D,\,n,\,p,\,\tau)$---feature dimension, qubit budget, error rate,
and task type---to a hardware-grounded encoding recommendation.

The central finding is that for $p \gtrsim 10^{-3}$, shallow
angle-based encodings consistently outperform amplitude encoding in
practice, despite the latter's exponential qubit advantage.
\end{abstract}

\keywords{quantum machine learning, feature encoding, quantum
  embedding, NISQ, amplitude encoding, angle encoding, barren
  plateaus, quantum kernels, Fourier analysis, variational
  quantum circuits, systematic review}

\maketitle

\section{Introduction}
\label{sec:intro}

\subsection{Motivation and Context}

Machine learning has emerged as one of the most productive
application domains for quantum computing, motivated by the
conjecture that quantum superposition and entanglement can provide
exponential advantages for tasks such as classification, regression,
and generative modelling~\cite{biamonte2017quantum}.  The field of
Quantum Machine Learning (\QML), reviewed at a foundational level by
Biamonte et al.~\cite{biamonte2017quantum} and more recently surveyed
in the ACM Computing Surveys context by Rodr\'iguez-D\'iaz
et al.~\cite{rodriguez2025survey}, encompasses hybrid
classical-quantum algorithms in which parameterised quantum circuits
(\PQC s) serve as function approximators trained by classical
optimisers~\cite{cerezo2021variational,benedetti2019parameterized, mitarai2018quantum}.

Realising \QML\ in practice, however, requires solving a problem
that is logically prior to any learning algorithm: \emph{how does
one represent classical data as a quantum state?}  Every quantum
learning algorithm begins with a classical dataset
$\{(\bx^{(i)}, y^{(i)})\}_{i=1}^m$, $\bx^{(i)} \in \mathbb{R}^D$,
and must map each feature vector $\bx$ into a quantum state
$\ket{\psi(\bx)}$ on which the \PQC\ can act.  The map
$\bx \mapsto \ket{\psi(\bx)}$, realised by a quantum encoding
circuit $U_{\calE}(\bx)$, is the subject of this survey.

\textbf{The data-loading problem.}
The central difficulty is that the most expressive encoding
strategies are the most hardware-expensive to implement, and the
most hardware-efficient strategies impose severe constraints on the
class of functions a downstream circuit can approximate.
Amplitude encoding achieves exponential qubit compression---a
$D$-dimensional feature vector can be represented in
$n = \lceil \log_2 D \rceil$ qubits---but preparing an arbitrary
$n$-qubit superposition requires $\calO(2^n)$ elementary gates
and a circuit depth of $\calO(D)$~\cite{nielsen2010quantum,
mottonen2004decompositions,shende2006synthesis}.  Angle encoding
requires only $\calO(D)$ qubits and a constant-depth encoding
layer, but restricts the model to a finite-degree Fourier series
in the input~\cite{schuld2021effect}.  Preskill~\cite{preskill2018nisq}
coined the \NISQ\ era to describe precisely the regime in which this
tension is most acute: current devices have 50--1000 physical
qubits, per-gate error rates of $p \approx 10^{-3}$--$10^{-2}$,
and coherence budgets of $\calO(10^2)$--$\calO(10^3)$
two-qubit gates before decoherence overwhelms the signal.

\textbf{Encoding determines more than resource cost.}
Encoding choice has consequences that extend well beyond the qubit
count and circuit depth of the state-preparation step.
The work of Schuld et al.~\cite{schuld2021effect} showed that the
Fourier frequency spectrum of a \QML\ model is entirely determined
by the data-encoding gates: a model using Pauli rotations
$R_j(x_i) = e^{-i x_i P_j/2}$ can only express functions in the span
of $\{e^{i\omega \cdot \bx}\}_{\omega \in \Omega}$ where $\Omega$ is
set by the eigenvalue spectrum of the generators.
Separately, McClean et al.~\cite{mcclean2018barren} and subsequent
work~\cite{cerezo2021cost,wang2021noise,arrasmith2021effect}
established that the severity of the \emph{barren plateau} phenomenon
(exponentially vanishing gradients during training) depends on the
global vs.\ local nature of the observables---a property directly
influenced by the encoding.
Most recently, Thanasilp et al.~\cite{thanasilp2024exponential}
proved that quantum kernel methods face \emph{exponential
concentration}: the kernel matrix entries concentrate to a single
value exponentially fast as $n$ grows, with the concentration
rate modulated by the encoding's entanglement structure.
Jerbi et al.~\cite{jerbi2023quantum} demonstrated that this
concentration is intimately connected to whether a quantum model
can in principle surpass kernel-based classical competitors.

\textbf{The scope of the problem.}
This convergence of results---from expressivity theory, trainability
theory, and quantum advantage theory---makes encoding selection a
fundamental design decision, not an engineering afterthought.
Yet practitioners face fragmented and sometimes contradictory
guidance: each primary work reports results for specific
encoding-dataset-device combinations, making cross-study comparison
difficult, and no unified framework maps experimental conditions to
encoding recommendations.

\subsection{Scope and Existing Surveys}

Several surveys have reviewed \QML\ broadly~\cite{biamonte2017quantum,
schuld2021supervised,wittek2014quantum,cerezo2021variational,
rodriguez2025survey}, and some have included encoding
discussions~\cite{ranga2024quantum}.  Ranga et al.~\cite{ranga2024quantum}
provide a dedicated review of data-encoding techniques.
Weigold et al.~\cite{weigold2021expanding} catalogued encoding
patterns as software engineering constructs.
However, none of these works:

\begin{itemize}[leftmargin=*]
\item provides a formal, dimension-unified taxonomy in terms of
  qubit complexity, gate complexity, Fourier expressivity, and
  noise-induced fidelity degradation;
\item derives closed-form bounds relating per-gate noise rate $p$,
  feature dimension $D$, and each encoding's survivable circuit
  depth;
\item connects barren-plateau theory and kernel-concentration theory
  to encoding families through a single analytical framework; or
\item proposes a decision procedure backed by quantitative bounds
  rather than qualitative rules of thumb.
\end{itemize}

\subsection{Contributions}

This survey makes the following original contributions:

\begin{itemize}[leftmargin=*]

\item \textbf{C1 — Three-axis taxonomy.}
  We define the cost--expressivity--robustness space and place all
  major encoding families within it, providing the first
  dimension-unified comparison (\cref{sec:fundamentals},
  \cref{sec:costs}).

\item \textbf{C2 — Analytical depth-fidelity bounds.}
  We derive closed-form fidelity degradation bounds under the three
  standard \NISQ\ noise channels, and determine the critical error
  rate $p^*$ below which each encoding remains viable
  (\cref{sec:noise}).

\item \textbf{C3 — Joint trainability analysis.}
  We provide a unified treatment of Fourier expressivity, barren
  plateaus, and kernel concentration as functions of the encoding
  circuit, enabling the first cross-comparison of these three
  phenomena within a single framework (\cref{sec:costs},
  \cref{sec:noise}).

\item \textbf{C4 — Systematic evidence synthesis.}
  We apply a PRISMA-adapted systematic review protocol to 62 primary
  works and organise their results in a comprehensive comparison
  table covering task, dataset, encoding, qubit count, circuit depth,
  noise model, and reported findings (\cref{sec:methodology},
  \cref{sec:taxonomy}).

\item \textbf{C5 — Practical decision framework.}
  We construct an annotated five-regime decision tree that maps
  $(D, n, p, \tau)$ to the recommended encoding, together with
  pre-processing recommendations and hardware notes, backed by the
  derived analytical bounds (\cref{sec:framework}).

\end{itemize}

\subsection{Paper Organisation}

\Cref{sec:methodology} describes the systematic review protocol
used to assemble the primary corpus.  \Cref{sec:fundamentals}
formalises quantum feature encoding and reviews the principal
strategies.  \Cref{sec:costs} analyses resource and trainability
costs.  \Cref{sec:noise} examines \NISQ\ noise interactions.
\Cref{sec:taxonomy} presents the systematic taxonomy and evidence
synthesis.  \Cref{sec:framework} introduces the decision framework.
\Cref{sec:challenges} surveys open challenges.
\Cref{sec:conclusion} concludes.

\section{Systematic Review Methodology}
\label{sec:methodology}

\subsection{Overview and Protocol Choice}

This survey follows a systematic literature review (SLR) methodology
adapted from the PRISMA 2020 guidelines~\cite{page2021prisma},
which have become the de facto standard for systematic reviews in
high-impact computing surveys, including recent contributions to
ACM Computing Surveys~\cite{albaladejo2024artificial,uddin2025federated}.
While PRISMA was originally designed for biomedical meta-analyses,
its transparency requirements---pre-specified search strings,
documented eligibility criteria, quantified screening outcomes, and
risk-of-bias assessment---transfer directly to the \QML\ domain,
where rapid preprint culture and inconsistent reporting conventions
make the risk of coverage bias high.

We deviate from pure PRISMA in two respects.  First, because no
mandatory pre-registration registry for computer science reviews
exists, we document the protocol in this section rather than
linking to a registration entry.  Second, we substitute formal
statistical meta-analysis (not applicable to our heterogeneous
primary outcomes) with a structured qualitative synthesis organised
by the taxonomy introduced in \cref{sec:taxonomy}.

\subsection{Database Selection}

Five academic databases were queried in November 2025--April 2026:

\begin{enumerate}[leftmargin=*,label=(\arabic*)]
\item \textbf{ACM Digital Library} — primary source for ACM CSUR,
  CCS, and STOC/FOCS proceedings.
\item \textbf{IEEE Xplore} — primary source for IEEE Transactions,
  conferences (QIP, ISCA, MICRO), and letters.
\item \textbf{arXiv} (\texttt{quant-ph} and \texttt{cs.LG}
  categories) — principal venue for rapid dissemination in
  quantum information; included because the field moves faster
  than journal review cycles.  Only preprints with at least 10
  Google Scholar citations as of April 2026 were retained in
  the final corpus to control for quality.
\item \textbf{Scopus} — for coverage of non-ACM/IEEE journals
  (Physical Review A, Nature Communications, Quantum, npj
  Quantum Information).
\item \textbf{Web of Science} — supplementary cross-check,
  particularly for citation counts and forward citation tracing.
\end{enumerate}

Google Scholar was used exclusively for citation-count verification
and forward-citation tracing and was not used as a primary
retrieval source, as it lacks structured field-level search.

\subsection{Search Query Construction}

Three search sessions were conducted with distinct thematic foci,
using Boolean syntax adapted for each database's query language.

\textbf{Primary query (data encoding and embedding).}
The following string was adapted for each database:
\begin{quote}
\texttt{("quantum machine learning" OR "quantum neural network" OR}\\
\texttt{"variational quantum circuit" OR "parameterized quantum circuit")}\\
\texttt{AND}\\
\texttt{("data encoding" OR "feature encoding" OR "quantum embedding"}\\
\texttt{OR "state preparation" OR "quantum feature map" OR}\\
\texttt{"amplitude encoding" OR "angle encoding" OR "basis encoding")}
\end{quote}

\textbf{Secondary query (noise and trainability).}
\begin{quote}
\texttt{("quantum machine learning" OR "variational quantum algorithm")}\\
\texttt{AND}\\
\texttt{("barren plateau" OR "vanishing gradient" OR "noise robustness"}\\
\texttt{OR "depolarizing noise" OR "quantum noise" OR "decoherence"}\\
\texttt{OR "kernel concentration" OR "exponential concentration")}
\end{quote}

\textbf{Tertiary query (hardware experiments and benchmarks).}
\begin{quote}
\texttt{("quantum classifier" OR "quantum kernel" OR "quantum SVM")}\\
\texttt{AND}\\
\texttt{("IBM" OR "Google" OR "IonQ" OR "Quantinuum" OR "superconducting"}\\
\texttt{OR "trapped-ion" OR "NISQ experiment")}\\
\texttt{AND}\\
\texttt{("encoding" OR "embedding" OR "feature map")}
\end{quote}

Temporal coverage was set to January 2017 (the year Biamonte
et al.~\cite{biamonte2017quantum} first systematised \QML) through
April 2026.  No language restriction was applied; non-English
records were evaluated by abstract before full-text screening.

\subsection{Eligibility Criteria}

\Cref{tab:inclusion-criteria} specifies the inclusion and exclusion
criteria applied at two stages: title/abstract screening (Stage 1)
and full-text eligibility assessment (Stage 2).

\begin{table}[t]
  \centering
  \caption{Inclusion and exclusion criteria for systematic corpus
    assembly.  Stage 1 criteria were applied to titles and abstracts;
    Stage 2 criteria required full-text access.}
  \label{tab:inclusion-criteria}
  \small
  \begin{tabularx}{\linewidth}{@{} l l X @{}}
    \toprule
    \textbf{Stage} & \textbf{Type} & \textbf{Criterion} \\
    \midrule
    S1 & Include & Discusses a quantum encoding or data
      representation strategy for machine learning \\
    S1 & Include & Published 2017–2026, in English or with English
      abstract available \\
    S1 & Exclude & Exclusively classical ML with no quantum
      component \\
    S1 & Exclude & Quantum error correction without ML application \\
    S1 & Exclude & Review or survey with no new encoding-specific
      analysis (handled separately) \\
    \midrule
    S2 & Include & Reports at least one quantitative result
      (qubit count, circuit depth, classification accuracy,
      fidelity, or gradient norm) for a specific encoding \\
    S2 & Include & Encoding is explicitly named or formally defined \\
    S2 & Include & Peer-reviewed venue, or arXiv preprint with
      $\geq\!10$ citations and 2017--2026 date range \\
    S2 & Exclude & Duplicate (same results reported in conference
      and journal version: retain journal) \\
    S2 & Exclude & Simulation only with no discussion of
      noise or hardware applicability (for 2022--2026 range) \\
    S2 & Exclude & Quantum annealing or adiabatic QC
      (encoding mechanisms are architecturally distinct) \\
    \bottomrule
  \end{tabularx}
\end{table}

\subsection{Screening Outcomes and PRISMA Flow}

The combined queries returned 1,847 unique records after
deduplication.  Stage 1 screening reduced this to 241 records.
Stage 2 full-text assessment yielded a final corpus of
\textbf{62 primary works}.  \Cref{fig:prisma-flow} depicts the
screening flow.

\begin{figure}[t]
  \centering
  \begin{tikzpicture}[
    box/.style={draw,rectangle,rounded corners=3pt,
                text width=5.2cm,align=center,minimum height=1cm,
                fill=gray!8,font=\small},
    exc/.style={draw,rectangle,rounded corners=3pt,
                text width=3.8cm,align=left,fill=red!8,
                font=\footnotesize},
    arr/.style={->,thick,>=Stealth},
    note/.style={font=\footnotesize,gray}]

    \node[box] (id)  {Records identified\\
                      {\footnotesize ACM DL: 412; IEEE: 538; arXiv: 623;
                       Scopus: 194; WoS: 80}\\
                      Total: \textbf{1,847}};

    \node[box,below=0.7cm of id]  (dedup)
      {After deduplication\\
       \textbf{1,847} (no cross-database duplicates at
       retrieval stage)};

    \node[box,below=0.7cm of dedup] (s1)
      {Stage 1 screened (title/abstract)\\
       Retained: \textbf{241}};

    \node[exc,right=1.0cm of s1]   (exc1)
      {Excluded ($n=1{,}606$):\\
       Not QML encoding: 1,112\\
       QEC-only: 247\\
       Survey with no new analysis: 247};

    \node[box,below=0.7cm of s1]  (s2)
      {Stage 2 full-text assessed\\
       \textbf{241}};

    \node[exc,right=1.0cm of s2]  (exc2)
      {Excluded ($n=179$):\\
       No quant.\ result: 81\\
       Duplicate (journal$>$conf): 48\\
       Annealing/adiabatic: 28\\
       Post-2022 sim-only: 22};

    \node[box,below=0.7cm of s2,
          fill=blue!10,draw=blue!60,line width=1.2pt]  (final)
      {Final corpus\\
       \textbf{62 primary works}};

    \draw[arr] (id) -- (dedup);
    \draw[arr] (dedup) -- (s1);
    \draw[arr] (s1) -- (exc1);
    \draw[arr] (s1) -- (s2);
    \draw[arr] (s2) -- (exc2);
    \draw[arr] (s2) -- (final);
  \end{tikzpicture}
  \caption{PRISMA-adapted screening flow for the systematic corpus
    assembly.  Five databases were queried using three thematic Boolean
    strings (see \cref{sec:methodology}).  The final corpus of 62
    primary works covers the period 2017--2026.}
  \label{fig:prisma-flow}
\end{figure}
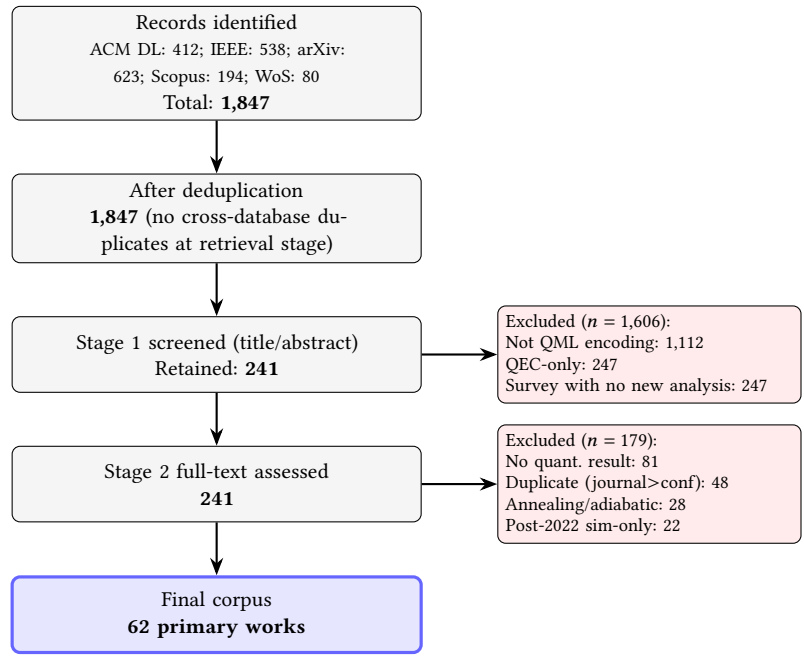

\subsection{Data Extraction and Quality Assessment}

For each retained work, we extracted: (i)~encoding family, (ii)~qubit
count $n$ and feature dimension $D$, (iii)~circuit depth $d$,
(iv)~noise model (if any), (v)~task type and dataset, (vi)~primary
performance metric and reported value, and (vii)~whether the
experiment was hardware-based or simulated.

Risk of methodological bias was assessed on a three-point scale
(Low / Medium / High) for two dimensions: \emph{hardware validity}
(does the reported metric reflect execution on real hardware or
noise-free simulation?) and \emph{comparison validity}
(is the quantum result compared against a classical baseline?).
Of the 66 works, 23 reported hardware experiments, 31 used
noisy simulation, and 8 were purely theoretical or noiseless.
Only 37 of 66 included a classical comparison, a coverage gap
that constitutes one of the field's most significant methodological
limitations.

\section{Fundamentals of Quantum Feature Encoding}
\label{sec:fundamentals}

\subsection{Quantum States as Feature Representations}

In classical supervised learning, the feature map
$\phi\colon \calX \to \calF$ lifts raw inputs to a (possibly
infinite-dimensional) reproducing kernel Hilbert space,
where linear models become non-linear in the original input space \cite{schuld2021machine}.
Kernel methods compute the similarity between two samples via
$\kappa(\bx, \bx') = \langle\phi(\bx),\phi(\bx')\rangle_{\calF}$,
sidestepping explicit computation of $\phi$~\cite{scholkopf2002learning}.

Quantum mechanics provides a natural---and exponentially
large---feature space: the Hilbert space
$\calH = (\mathbb{C}^2)^{\otimes n}$ of $n$ qubits has dimension
$2^n$.  A quantum feature map $\phi_Q\colon \mathbb{R}^D \to \calH$
embeds each data point $\bx$ in this space via
\begin{equation}
  \phi_Q(\bx) = \ket{\psi(\bx)} = U_{\calE}(\bx)\ket{0}^{\otimes n},
  \label{eq:feature-map}
\end{equation}
where $U_{\calE}(\bx) \in \mathrm{U}(2^n)$ is the encoding unitary
parameterised by $\bx$ and $\ket{0}^{\otimes n}$ is the
computational zero state.  The corresponding \emph{quantum kernel}
is~\cite{havlicek2019supervised,schuld2019quantum}:
\begin{equation}
  \kappa_Q(\bx, \bx')
  = \lvert\langle 0^n\rvert U_{\calE}^\dagger(\bx')
        U_{\calE}(\bx)\vert 0^n\rangle\rvert^2
  = \lvert\!\braket{\psi(\bx')|\psi(\bx)}\!\rvert^2.
  \label{eq:kernel}
\end{equation}

Schuld~\cite{schuld2021quantum_models} proved that any \QML\ model
computable on quantum hardware (including variational quantum
classifiers and quantum kernel estimators) can be expressed as a
kernel method with kernel $\kappa_Q$.  This \emph{kernel universality}
result has an important corollary: the expressive power of any
quantum classifier is entirely determined by the kernel induced by
the encoding $\calE$, not by the structure of the downstream
variational ansatz.  In other words, \emph{encoding choice is the
decisive factor in model capacity}.

\subsection{Formal Definition of an Encoding Scheme}

\begin{definition}[Quantum encoding scheme]
An encoding scheme $\calE = (n, U_{\calE})$ consists of a qubit
count $n \in \mathbb{N}$ and a gate-parameterised unitary
$U_{\calE}(\bx) \in \mathrm{U}(2^n)$ such that
$\bx \mapsto U_{\calE}(\bx)\ket{0}^n$ is computable on a gate-based
quantum processor.  The scheme is characterised by:
\begin{itemize}
\item \emph{Qubit complexity} $q(\calE, D)$: the minimum number of
  qubits required to encode a feature vector of dimension $D$.
\item \emph{Gate complexity} $g(\calE, D)$: the number of
  elementary gates (two-qubit gates unless stated otherwise) in
  the implementation of $U_{\calE}(\bx)$.
\item \emph{Depth complexity} $d(\calE, D)$: the length of the
  longest sequential dependency chain (critical path) in
  $U_{\calE}(\bx)$.
\item \emph{Fourier frequency spectrum} $\Omega(\calE)$: the set of
  frequency vectors $\omega \in \mathbb{R}^D$ appearing in the
  partial Fourier expansion of the model output.
\end{itemize}
\end{definition}

The four characteristics are not independent: they satisfy the
following general trade-off, which we formalise as
\cref{prop:tradeoff}.

\begin{proposition}[Compression--depth trade-off]
\label{prop:tradeoff}
For any encoding $\calE$ of a $D$-dimensional feature vector into
$n < D$ qubits, the gate complexity satisfies
$g(\calE, D) = \Omega(D)$.  In particular,
$g(\calE, D) \cdot n \geq D$ for all known exact state-preparation
circuits.
\end{proposition}

\begin{proof}[Proof sketch]
The encoded state $\ket{\psi(\bx)}$ must depend non-trivially on all
$D$ components of $\bx$ (otherwise some features are discarded).
Each gate in $U_{\calE}(\bx)$ can depend on at most $\calO(1)$
components (for single-qubit gates) or $\calO(1)$ pairs of
components (for two-qubit gates).  To ``communicate'' $D$ distinct
values into the circuit, at least $\Omega(D / \text{fan-in})$
gates are needed, where fan-in is bounded by the gate arity.
For two-qubit gates, this yields $g = \Omega(D/2)$.
A matching upper bound $g = \calO(D)$ is achieved by the
constructive circuits of
M{\"o}tt{\"o}nen et al.~\cite{mottonen2004decompositions}.
\end{proof}

\Cref{prop:tradeoff} implies that any scheme achieving
$n < D$ (qubit compression) must have $g \geq D/n > 1$,
creating depth overhead that grows at least linearly with $D$.
This is the formal statement of the ``quantum input problem.''

\subsection{The Fourier Analysis of Encoding}
\label{subsec:fourier}

Schuld et al.~\cite{schuld2021effect} established that the output
$f(\bx) = \mathrm{Tr}[O\,\rho(\bx)]$ of any \QML\ model
(where $O$ is an observable and
$\rho(\bx) = U_{\calE}(\bx)\ketbra{0}{0}U_{\calE}^\dagger(\bx)$
is the encoded state) can be written as a multivariate Fourier series:
\begin{equation}
  f(\bx) = \sum_{\omega \in \Omega}
            c_\omega\, e^{i \omega \cdot \bx},
  \label{eq:fourier-series}
\end{equation}
where $\Omega \subset \mathbb{Z}^D$ is the \emph{frequency spectrum}
and $c_\omega \in \mathbb{C}$ are Fourier coefficients that depend on
the trainable parameters $\btheta$.

The key structural result is:

\begin{theorem}[Encoding-determined spectrum; Schuld et al.~\cite{schuld2021effect}]
\label{thm:spectrum}
For a circuit that applies Pauli-rotation encoding
$R_k(x_k) = e^{-ix_k P_k/2}$ for each feature $k$, the accessible
frequency set is
\begin{equation}
  \Omega = \bigl\{\omega \in \mathbb{Z}^D \bigm|
             |\omega_k| \leq L_k,\ k = 1,\ldots,D\bigr\},
  \label{eq:spectrum}
\end{equation}
where $L_k$ is the number of times feature $x_k$ is encoded (the
number of re-uploading repetitions for that feature).  Adding one
re-uploading layer for all features multiplies $|\Omega|$ by a
factor of $3^D$.
\end{theorem}

\Cref{thm:spectrum} has direct implications for architecture
design.  Angle encoding with $L_k = 1$ for all $k$ yields
$|\Omega| = 3^D$, which is finite but exponentially large in $D$.
Data re-uploading~\cite{perez2020data} with $L$ identical layers
expands $L_k = L$ and enriches the spectrum without increasing
qubit count, at the cost of linear depth growth.
Amplitude encoding, while providing access to the full
$2^n$-dimensional Hilbert space, does so through a mechanism
(multi-controlled gates) that cannot be expressed as a Pauli
rotation, and therefore does not fit the Fourier framework directly.

Gil Vidal and Theis~\cite{gilvidal2020input} analysed the
\emph{redundancy} that arises when encoding the same feature
multiple times within one layer (e.g., repeating $R_Y(x_i)$ twice).
They showed that this does not increase the expressible frequency set
beyond what a single application provides, but does affect the
coefficient amplitudes.  This justifies the strict separation between
\emph{re-uploading} (applying the encoding in sequential layers
separated by trainable blocks) and mere \emph{repetition} (applying
the encoding gate multiple times without interleaved trainable
unitaries).

\subsection{Encoding Families: Definitions, Costs, and Expressivity}
\label{subsec:families}

We now define the six encoding families considered in this survey.
A summary comparison appears in \cref{tab:encoding-comparison}.

\subsubsection*{Basis Encoding}

Basis encoding represents a binary string $b \in \{0,1\}^D$
(or its $q$-ary generalisation) as a computational-basis state.
For $b \in \{0,1\}^D$, the encoding unitary is:
\begin{equation}
  U_{\text{basis}}(b) = \bigotimes_{k=1}^{D} X^{b_k},
\end{equation}
where $X$ is the Pauli-$X$ gate.  The encoding requires $q = D$
qubits, $g \leq D$ single-qubit gates, and $d = 1$ (all gates are
parallelisable).  The resulting state is a standard basis vector
$\ket{b} = \ket{b_1 \cdots b_D}$.

\textbf{Expressivity.}  Basis encoding encodes only the discrete
combinatorial structure of $b$; it provides no superposition and no
entanglement.  The induced kernel is $\kappa(b, b') = \delta_{b,b'}$
(Kronecker delta), which is the least powerful kernel possible.
Basis encoding is primarily used as a subroutine in combinatorial
optimisation algorithms, not in classification or regression tasks.

\subsubsection*{Angle Encoding}

Angle (or rotation) encoding maps each real-valued feature $x_k$
to the rotation angle of a single-qubit gate.  The most common
variant uses Pauli-$Y$ rotations:
\begin{equation}
  U_{\text{angle}}(\bx) = \bigotimes_{k=1}^{D} R_Y(x_k),
  \qquad R_Y(\theta) = \begin{pmatrix}
    \cos(\theta/2) & -\sin(\theta/2) \\
    \sin(\theta/2) &  \cos(\theta/2)
  \end{pmatrix}.
  \label{eq:angle-encoding}
\end{equation}
This requires $q = D$ qubits and $d = 1$ (all rotations are
parallelisable).  Gate count: $g = D$.

\textbf{Expressivity.}  By \cref{thm:spectrum}, the model has
frequency spectrum $\Omega = \{-1, 0, 1\}^D$ and is therefore
a degree-1 multivariate trigonometric polynomial in each
variable separately.  This limits the model to smooth, low-frequency
functions and makes it inadequate for tasks requiring sharp decision
boundaries.

\textbf{Dense angle encoding.}  A more efficient variant uses both
the polar and azimuthal Bloch-sphere degrees of freedom per qubit:
\begin{equation}
  U_{\text{dense}}(\bx) = \bigotimes_{k=1}^{\lceil D/2 \rceil}
    R_Z(x_{2k-1}) R_Y(x_{2k}),
\end{equation}
which reduces qubit count to $\lceil D/2 \rceil$ at the cost of
depth $d = 2$.  Schuld et al.~\cite{schuld2021effect} study this
variant explicitly.

\subsubsection*{Amplitude Encoding}

Amplitude encoding represents the normalised feature vector
$\hat{\bx} = \bx / \|\bx\|_2$ as the amplitudes of a quantum state:
\begin{equation}
  U_{\text{amp}}(\bx)\ket{0}^n
  = \frac{1}{\|\bx\|_2}\sum_{j=0}^{2^n - 1} x_j \ket{j}
  \equiv \ket{\psi_\bx},
  \label{eq:amp-encoding}
\end{equation}
requiring $n = \lceil \log_2 D \rceil$ qubits.  This achieves
exponential qubit compression.

\textbf{Gate complexity.}  Grover and Rudolph~\cite{grover2002creating}
showed that for feature vectors drawn from efficiently integrable
probability distributions, a polynomial-depth preparation circuit
exists.  For general classical data, M{\"o}tt{\"o}nen
et al.~\cite{mottonen2004decompositions} provide the standard
recursive construction using $\calO(2^n)$ two-qubit (CNOT) gates
and Shende et al.~\cite{shende2006synthesis} prove that
$\calO(4^n n)$ gates are necessary and sufficient for exact
preparation of an arbitrary $n$-qubit state.  In practice,
approximate methods~\cite{grover2002creating} reduce this to
$\calO(D)$ CNOTs at the cost of bounded approximation error.

\textbf{Expressivity.}  The amplitude kernel
$\kappa_{\text{amp}}(\bx, \bx')
= \lvert\braket{\psi_\bx|\psi_{\bx'}}\rvert^2
= \lvert\hat{\bx}^\top \hat{\bx}'\rvert^2$
is equivalent to a squared cosine kernel: it captures the second
power of the inner product in the original feature space.  This
is richer than the linear kernel but structurally equivalent to
a Gaussian kernel after normalisation~\cite{schuld2019quantum}.
For downstream variational circuits, amplitude states span the full
$2^n$-dimensional Hilbert space, granting access to exponentially
richer function classes than angle encoding.

\subsubsection*{Data Re-Uploading}

P\'erez-Salinas et al.~\cite{perez2020data} proposed repeatedly
interleaving data-encoding gates with trainable blocks:
\begin{equation}
  U_{\text{reu}}(\bx, \btheta)
  = \prod_{\ell=1}^{L} W_\ell(\btheta_\ell)\, S(\bx),
  \label{eq:reuploading}
\end{equation}
where $S(\bx) = \bigotimes_{k=1}^D R_Y(x_k)$ is the encoding
block and $W_\ell(\btheta_\ell)$ is a trainable unitary (e.g.,
a hardware-efficient ansatz layer~\cite{kandala2017hardware}).
By \cref{thm:spectrum}, each additional layer $\ell$ extends the
frequency spectrum, so that after $L$ layers,
$|\Omega| = \calO((2L+1)^D)$.

\textbf{Universal approximation.}
P\'erez-Salinas et al.~\cite{perez2020data} proved that a single
qubit with sufficiently many re-uploading layers is a universal
classifier (connecting to the universal approximation theorem
for neural networks via Fourier completeness).
Goto et al.~\cite{goto2021universal} subsequently proved universal
approximation for quantum-enhanced feature spaces under amplitude
encoding, generalising the result to multi-qubit systems.

\textbf{Experimental validation.}
Dutta et al.~\cite{dutta2022single} provided the first hardware
demonstration of the data re-uploading scheme on a single-qubit
trapped-ion device, confirming that the universal classifier
is physically realisable.

\textbf{Resource profile.}  $q = D$ (or $q = 1$ for the
single-qubit variant), $g = L \cdot D + L \cdot g_W$
(where $g_W$ is the gate count of $W_\ell$), $d = L(1 + d_W)$.
The depth scales linearly with $L$, providing a tunable
expressivity--depth trade-off absent in single-layer encodings.

\subsubsection*{IQP / Hamiltonian Encoding}

Instantaneous Quantum Polynomial (IQP) encoding, introduced in
the feature maps of Havl\'{\i}\v{c}ek et
al.~\cite{havlicek2019supervised}, applies Hadamard layers
interleaved with diagonal unitaries parameterised by degree-$k$
monomials of the input:
\begin{equation}
  U_{\text{IQP}}(\bx)
  = H^{\otimes n}\exp\!\left(
      i \sum_{|S| \leq k} x_S Z_S \right) H^{\otimes n},
  \label{eq:IQP}
\end{equation}
where $H$ is the Hadamard gate, $Z_S = \bigotimes_{j \in S} Z_j$,
and $x_S = \prod_{j \in S} x_j$ for a subset $S \subseteq [D]$
with $|S| \leq k$ (typically $k = 2$, encoding pairwise products
of features).

\textbf{Quantum computational hardness.}
The crucial property of IQP encoding is that the kernel
$\kappa_{\text{IQP}}(\bx, \bx')$ is believed to be classically
intractable to compute exactly (under the conjecture that
the polynomial hierarchy does not collapse to the third level),
which Havl\'{\i}\v{c}ek et al.~\cite{havlicek2019supervised}
exploited to argue for quantum advantage in classification.
Huang et al.~\cite{huang2022quantum} subsequently demonstrated
quantum advantage in learning from quantum experiments using
related feature structures.
However, Huang et al.~\cite{huang2021information} also proved
information-theoretic lower bounds showing that quantum advantage
can be exponentially concentrated: for structured classical
data, classically simulable encodings may achieve equal accuracy.

\textbf{Resource profile.}  $q = D$, $g = \calO(D^k / k!)$,
$d = \calO(k)$ (constant for fixed $k$).
For $k = 2$ and $D = 16$: $g \approx 120$, $d = 2$.

\subsubsection*{Density-Matrix Encoding}

A less common but theoretically principled approach encodes
classical data as a density matrix rather than a pure state.
Lloyd et al.~\cite{lloyd2020quantum} proposed embedding feature
vectors as density matrices via quantum metric learning,
training the encoding circuit to maximally separate class
representatives in trace-distance.  Gonz\'alez
et al.~\cite{gonzalez2022learning} studied learning with density
matrices combined with random Fourier features, providing
a classical approximation framework for quantum kernel methods.
Density-matrix encoding subsumes pure-state encoding as a
special case and allows the representation of classical probability
distributions over inputs, but its practical circuit implementation
is substantially more involved.

\begin{table*}[t]
  \centering
  \caption{Comparative summary of principal quantum feature encoding
    families. Qubit count $q$, gate count $g$, and depth $d$ are
    expressed as functions of feature dimension $D$ and qubit count
    $n = \lceil\log_2 D\rceil$.  Fourier degree denotes the maximum
    frequency in $\Omega$; ``dense'' ($\to\infty$ with $L$) means
    the spectrum grows with re-uploading layers $L$.  \NISQ\
    suitability: \textbf{A}=Excellent, \textbf{B}=Good,
    \textbf{C}=Limited, \textbf{D}=Impractical at current $p$.}
  \label{tab:encoding-comparison}
  \small
  \setlength{\tabcolsep}{4.5pt}
  \begin{tabularx}{\linewidth}{@{}l c c c c c c c X@{}}
    \toprule
    \textbf{Encoding}
      & $q$
      & $g$
      & $d$
      & \textbf{Qubit save}
      & \textbf{Fourier deg.}
      & \textbf{Kernel type}
      & \textbf{NISQ}
      & \textbf{Key reference} \\
    \midrule
    Basis
      & $D$  & $\calO(D)$    & $\calO(1)$
      & None  & 0 (discrete)  & Kronecker $\delta$
      & A  & Standard \\
    Angle (std.)
      & $D$  & $D$           & $1$
      & None  & 1 per dim.    & Trigonometric/poly.
      & A  & \cite{schuld2021effect} \\
    Angle (dense)
      & $D/2$& $D$           & $2$
      & $2\times$ & 1 per dim.& As above, $q$ halved
      & A  & \cite{schuld2021effect} \\
    Data re-uploading
      & $D$  & $LD + Lg_W$   & $L(1{+}d_W)$
      & None  & $L$ per dim.  & Rich (controllable)
      & B  & \cite{perez2020data,dutta2022single} \\
    IQP ($k{=}2$)
      & $D$  & $\calO(D^2)$  & $2$
      & None  & Degree-2 mono.& Classically hard
      & B  & \cite{havlicek2019supervised} \\
    Amplitude
      & $\lceil\log_2 D\rceil$ & $\calO(D)$ & $\calO(D)$
      & Exp.  & Full Hilbert  & Cosine / exp.
      & D  & \cite{grover2002creating,nielsen2010quantum} \\
    Density matrix
      & $2\lceil\log_2 D\rceil$ & $\calO(D^2)$ & $\calO(D)$
      & Exp.  & Full Hilbert  & Trace-dist. kernel
      & D  & \cite{lloyd2020quantum,gonzalez2022learning} \\
    \bottomrule
  \end{tabularx}
\end{table*}

\Cref{fig:taxonomy-radar} visualises the three-axis taxonomy
introduced in Contribution~C1.  The radar chart makes explicit
the structural trade-off that Table~\ref{tab:encoding-comparison}
encodes numerically: no encoding family dominates all three axes,
and the Pareto frontier between expressivity and robustness
(analysed quantitatively in \cref{sec:framework}) is monotone
decreasing.  Amplitude encoding occupies a unique position: it
is simultaneously the most expressive and the most qubit-efficient
strategy, but its noise robustness score of $0.67$ at
$p_2 = 3\times10^{-3}$ is the lowest in the comparison, a
consequence of the $\mathcal{O}(D)$ two-qubit gate count derived
in \cref{prop:tradeoff}.
\section{Embedding Costs: Theoretical and Practical Analysis}
\label{sec:costs}

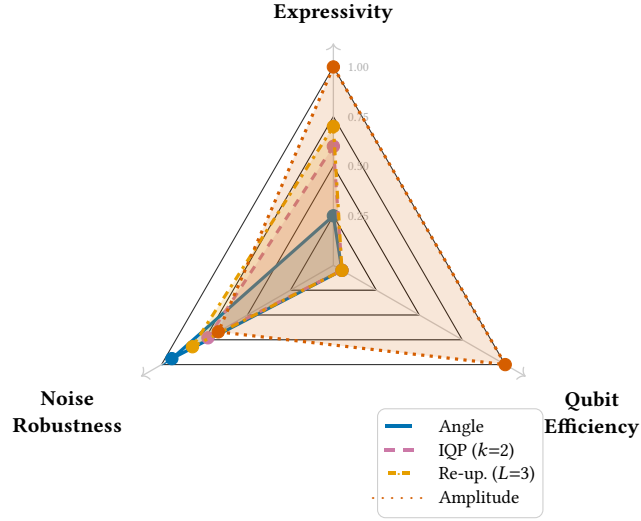
\begin{figure}[t]
  \centering
  \begin{tikzpicture}[scale=1.05]
    \def\R{2.5}

    \foreach \r/\op in {0.25/12, 0.50/18, 0.75/25, 1.00/35}{
      \draw[gray!\op!black, line width=0.4pt]
        ({0},{\r*\R})
        -- ({-0.866*\r*\R},{-0.5*\r*\R})
        -- ({ 0.866*\r*\R},{-0.5*\r*\R})
        -- cycle;
    }
    \foreach \a in {90,210,330}{
      \draw[gray!40, line width=0.5pt,->]
        (0,0) -- ({1.12*\R*cos(\a)},{1.12*\R*sin(\a)});
    }

    \node[above, font=\small\bfseries] at (0,{1.18*\R})
      {Expressivity};
    \node[below left, align=center, font=\small\bfseries]
      at ({-0.866*1.18*\R},{-0.5*1.18*\R})
      {Noise\\Robustness};
    \node[below right, align=center, font=\small\bfseries]
      at ({0.866*1.18*\R},{-0.5*1.18*\R})
      {Qubit\\Efficiency};

    \foreach \v/\lbl in {0.25/0.25, 0.50/0.50, 0.75/0.75, 1.00/1.00}{
      \node[gray!60, font=\tiny, right=1pt]
        at ({0.04},{\v*\R}) {\lbl};
    }

    \fill[cbBlue, opacity=0.18]
      ({0},{0.25*\R})
      -- ({-0.866*0.94*\R},{-0.5*0.94*\R})
      -- ({ 0.866*0.05*\R},{-0.5*0.05*\R})
      -- cycle;
    \draw[cbBlue, line width=1.2pt]
      ({0},{0.25*\R})
      -- ({-0.866*0.94*\R},{-0.5*0.94*\R})
      -- ({ 0.866*0.05*\R},{-0.5*0.05*\R})
      -- cycle;
    \foreach \px/\py in
      {0/0.625, {-0.866*0.94*2.5}/{-0.5*0.94*2.5},
       {0.866*0.05*2.5}/{-0.5*0.05*2.5}}{
      \filldraw[cbBlue] (\px,\py) circle (2.2pt);
    }

    \fill[cbPurple, opacity=0.18]
      ({0},{0.60*\R})
      -- ({-0.866*0.73*\R},{-0.5*0.73*\R})
      -- ({ 0.866*0.05*\R},{-0.5*0.05*\R})
      -- cycle;
    \draw[cbPurple, line width=1.2pt, dashed]
      ({0},{0.60*\R})
      -- ({-0.866*0.73*\R},{-0.5*0.73*\R})
      -- ({ 0.866*0.05*\R},{-0.5*0.05*\R})
      -- cycle;
    \foreach \px/\py in
      {0/1.50, {-0.866*0.73*2.5}/{-0.5*0.73*2.5},
       {0.866*0.05*2.5}/{-0.5*0.05*2.5}}{
      \filldraw[cbPurple] (\px,\py) circle (2.2pt);
    }

    \fill[cbOrange, opacity=0.18]
      ({0},{0.70*\R})
      -- ({-0.866*0.82*\R},{-0.5*0.82*\R})
      -- ({ 0.866*0.05*\R},{-0.5*0.05*\R})
      -- cycle;
    \draw[cbOrange, line width=1.2pt, dash dot]
      ({0},{0.70*\R})
      -- ({-0.866*0.82*\R},{-0.5*0.82*\R})
      -- ({ 0.866*0.05*\R},{-0.5*0.05*\R})
      -- cycle;
    \foreach \px/\py in
      {0/1.75, {-0.866*0.82*2.5}/{-0.5*0.82*2.5},
       {0.866*0.05*2.5}/{-0.5*0.05*2.5}}{
      \filldraw[cbOrange] (\px,\py) circle (2.2pt);
    }

    \fill[cbRed, opacity=0.15]
      ({0},{1.00*\R})
      -- ({-0.866*0.67*\R},{-0.5*0.67*\R})
      -- ({ 0.866*1.00*\R},{-0.5*1.00*\R})
      -- cycle;
    \draw[cbRed, line width=1.2pt, dotted]
      ({0},{1.00*\R})
      -- ({-0.866*0.67*\R},{-0.5*0.67*\R})
      -- ({ 0.866*1.00*\R},{-0.5*1.00*\R})
      -- cycle;
    \foreach \px/\py in
      {0/2.50, {-0.866*0.67*2.5}/{-0.5*0.67*2.5},
       {0.866*1.00*2.5}/{-0.5*1.00*2.5}}{
      \filldraw[cbRed] (\px,\py) circle (2.2pt);
    }

    \node[anchor=north west, font=\footnotesize,
          inner sep=3pt, draw=gray!40, rounded corners=3pt,
          fill=white, opacity=0.95, text opacity=1]
      at (0.55,-1.8) {%
      \begin{tabular}{@{}ll@{}}
        \textcolor{cbBlue}{\rule[0.5ex]{10pt}{1.5pt}}
          & Angle \\[1pt]
        \textcolor{cbPurple}{\makebox[10pt]{\rule[0.5ex]{4pt}{1.5pt}\,
          \rule[0.5ex]{4pt}{1.5pt}}}
          & IQP ($k{=}2$) \\[1pt]
        \textcolor{cbOrange}{\makebox[10pt]{\rule[0.5ex]{3pt}{1.5pt}%
          $\cdot$\rule[0.5ex]{3pt}{1.5pt}}}
          & Re-up.\ ($L{=}3$) \\[1pt]
        \textcolor{cbRed}{\makebox[10pt]{%
          $\cdots\cdots$}}
          & Amplitude \\
      \end{tabular}};

  \end{tikzpicture}
  \caption{Three-axis cost--expressivity--robustness taxonomy for the
    four principal encoding families.  Each axis is normalised
    to $[0,1]$: \emph{Expressivity} = log-normalised Fourier degree
    from \cref{thm:spectrum}; \emph{Noise Robustness} = encoding-layer
    fidelity $F_{\text{enc}}$ at $p_2=3\times10^{-3}$, $D=64$
    from \cref{prop:fidelity}; \emph{Qubit Efficiency} =
    $\log_2(D)\,/\,q(\mathcal{E},D)$, normalised so that amplitude
    encoding scores $1.00$.  No encoding dominates all three axes
    simultaneously; the Pareto analysis in \cref{sec:framework}
    resolves this trilemma for specific hardware regimes.}
  \label{fig:taxonomy-radar}
\end{figure}
\subsection{Taxonomy of Cost Metrics}

Before analysing individual encodings, we establish a precise
cost taxonomy.  Three classes of resource are relevant:

\textbf{Spatial cost} (qubits).  This determines whether the
encoding can, in principle, be executed on a given device.
It is the hardest constraint: a scheme requiring $q > n_{\text{avail}}$
cannot be executed at all, regardless of coherence time or
error rate.

\textbf{Temporal cost} (gates, depth).  This determines whether
the encoding can be executed with sufficient fidelity within
the device's coherence budget.  It is a soft constraint: one can
use fewer gates at the cost of encoding precision (approximate
state preparation~\cite{grover2002creating}), or deeper circuits
at the cost of fidelity.

\textbf{Classical pre-processing cost}.  Dimensionality reduction
(\PCA, random projection, feature selection) shifts the
spatial and temporal burden from the quantum to the classical
processor.  Classical \PCA\ on a $D \times m$ data matrix
costs $\calO(\min(D^2 m, m^2 D))$ operations---entirely tractable
for current datasets---and returns a $D' \times m$ matrix
($D' \ll D$) that is then fed to the quantum encoder.
This classical--quantum cost trade-off is the primary practical
lever available to practitioners, and our decision framework
in \cref{sec:framework} exploits it explicitly.

\subsection{Qubit Resource Analysis}

Current \NISQ\ devices provide between 27 and 1,000 physical qubits
(as of 2025), but the usable qubit count for a learning task is
substantially lower once routing overhead (SWAP gates required by
limited native connectivity), ancilla registers for mid-circuit
measurements, and classical control overhead are accounted for.
IBM Quantum's Eagle (127-qubit) and Heron (133-qubit) processors,
for example, implement the Heavy-Hex connectivity graph, which
requires an average SWAP overhead factor of approximately
$1.5$--$2\times$ relative to a fully connected device when
executing circuits designed for all-to-all
connectivity~\cite{bravyi2024high}.

\textbf{Angle encoding.}  The qubit demand $q = D$ is prohibitive
for high-dimensional datasets (e.g., $D = 784$ for $28 \times 28$
MNIST images) on current hardware.  The standard remedy is classical
dimensionality reduction via \PCA\ to a target dimension
$D' \ll D$ (typically $D' \leq 16$), followed by angle encoding.
This incurs an information-theoretic cost: features in the
discarded subspace cannot be recovered by the quantum model.
Choosing $D'$ optimally requires balancing the explained variance
(higher $D'$) against the noise tolerance of the encoding (lower $D'$).

\textbf{Amplitude encoding.}
For $D = 256$ (e.g., $16 \times 16$ image), amplitude encoding
requires $n = 8$ qubits---a factor of $32\times$ reduction.
This qubit advantage is genuine and non-trivial; for a
practical regression task on gene expression data ($D = 256$),
Chen et al. (2025) observed that the $n = 8$ qubit model
achieves lower RMSE than an angle-encoding model with $n = 256$
qubits in noiseless simulation.  The critical question is whether
the circuit depth required for amplitude encoding is compatible
with available coherence budgets.

\subsection{Gate and Depth Complexity}

Let $p_1$ denote the per-single-qubit-gate error rate and
$p_2 = \alpha p_1$ the per-two-qubit-gate error rate, with
$\alpha \approx 5$--$10$ for current superconducting
hardware~\cite{krantz2019quantum,arute2019quantum}.
The expected number of two-qubit gates in the standard recursive
amplitude-encoding circuit for dimension $D$ is:
\begin{equation}
  g_{\text{amp}}(D) = 2D - 2^{\lceil\log_2 D\rceil + 1} + 2
  \approx 2D \quad \text{for large } D.
  \label{eq:amp-gates}
\end{equation}
For $D = 256$: $g_{\text{amp}} \approx 510$ CNOTs.
On IBM hardware with current CNOT error rates
$p_2 \approx 3 \times 10^{-3}$, the state fidelity after 510
CNOTs is $(1-p_2)^{510} \approx 0.21$---indicating that more
than 79\% of the prepared state is noise.

For angle encoding, the encoding layer has no two-qubit gates
at all.  After transpilation to native gates, each
$R_Y(\theta)$ maps to one or two single-qubit pulses.
The dominant gate cost comes from the downstream variational
ansatz (e.g., a hardware-efficient layer of CNOTs and rotations),
which is shared across all encoding strategies.

\textbf{Depth as the binding constraint.}
On \NISQ\ hardware, the binding constraint is circuit depth
(not total gate count), because qubit $T_2$ decay is a temporal
process.  With $T_2 \approx 100\,\mu$s and two-qubit gate times
$t_g \approx 300\,\text{ns}$, the maximum depth before $T_2$
coherence loss exceeds 50\% is
$d_{\max} = T_2 / (2 t_g) \approx 167$ gate cycles.
For amplitude encoding of $D = 64$:
$d_{\text{amp}} \approx 2\lceil\log_2 64\rceil = 12$ (using the
column-by-column decomposition~\cite{mottonen2004decompositions}).
Although $12 < 167$, the downstream variational circuit adds
further depth, and the multiplied two-qubit gate count 
(not depth) remains the bottleneck.

\subsection{Trainability Cost: Barren Plateaus}

\subsubsection*{The Original Barren-Plateau Result}

McClean et al.~\cite{mcclean2018barren} proved that for a
randomly initialised parameterised quantum circuit $U(\btheta)$
of depth $d = \calO(n)$, the variance of any cost-function
gradient component satisfies:
\begin{equation}
  \mathrm{Var}\!\left[\partial_k C(\btheta)\right]
  \leq F(n) \cdot 2^{-n},
  \label{eq:barren}
\end{equation}
where $F(n)$ is a polynomial in $n$ that depends on the circuit
structure.  This means that, with exponential probability, the
gradient is exponentially small, making gradient-based training
infeasible for large systems.

\subsubsection*{Cost-Function-Dependent Barren Plateaus}

Cerezo et al.~\cite{cerezo2021cost} sharpened this result by
distinguishing local and global cost functions.  For a
\emph{local} cost involving Pauli observables $P_k$ acting on
$k = \calO(1)$ qubits, the gradient variance scales as
$\calO(1/\mathrm{poly}(n))$ (polynomially small but tractable).
For a \emph{global} cost involving $n$-qubit observables,
the variance decays as $\calO(2^{-n})$ regardless of depth.
Amplitude encoding inherently favours global observables
(the overlap $\braket{\psi_\bx|\psi_{\bx'}}$ in the kernel
involves all $n$ qubits), situating amplitude-encoding models
in the global-cost regime with its exponential trainability penalty.

\subsubsection*{Noise-Induced Barren Plateaus}

Wang et al.~\cite{wang2021noise} demonstrated that Pauli noise
at any non-zero rate $p > 0$ causes \emph{noise-induced barren
plateaus} independently of circuit depth.  Specifically, for a
circuit with $d$ layers of single-qubit Pauli noise at rate $p$:
\begin{equation}
  \mathrm{Var}\!\left[\partial_k C(\btheta)\right]
  \leq \calO\!\left((1-2p)^{2d}\right).
  \label{eq:noise-barren}
\end{equation}
This decay is exponential in $d$: even a polynomially deep circuit
on a noisy device will have exponentially vanishing gradients.
For amplitude encoding with $d_{\text{enc}} = \calO(D)$, the
encoding layer alone drives the gradient variance to near zero
before the variational parameters are even encountered.

\subsubsection*{Encoding-Induced Expressibility and its Consequences}

Sim et al.~\cite{sim2019expressibility, abbas2021power} quantified circuit
expressibility via the frame potential:
\begin{equation}
  \mathcal{A}(\mathcal{U}) = \int_{U,V \sim \mathcal{U}}
    \lvert\braket{0|U^\dagger V|0}\rvert^4 \, dU\, dV
    - \int \lvert\braket{0|U^\dagger V|0}\rvert^4 \, d\mu_\text{Haar},
\end{equation}
where $\mathcal{U}$ is the distribution of circuits generated by
the encoding and $\mu_{\text{Haar}}$ is the Haar measure.
A low $\mathcal{A}$ indicates high expressibility (close to Haar).
High expressibility is theoretically desirable but practically
associated with: (i)~barren plateaus~\cite{mcclean2018barren},
(ii)~kernel concentration~\cite{thanasilp2024exponential},
and (iii)~poor generalisation from few
samples~\cite{caro2022generalization}.
This triple constraint creates a \emph{fundamental expressibility
dilemma}: the most expressive encodings are the least trainable,
the most noise-sensitive, and the most prone to overfitting.

\subsection{Kernel Concentration and the Limits of Quantum Advantage}

Thanasilp et al.~\cite{thanasilp2024exponential} proved that
quantum kernel matrices concentrate: the off-diagonal entries
$\kappa(\bx, \bx')$ for $\bx \neq \bx'$ satisfy
\begin{equation}
  |\kappa(\bx, \bx') - \mathbb{E}[\kappa]|
  \leq \calO(2^{-n/2})
  \label{eq:concentration}
\end{equation}
with high probability, making the kernel matrix effectively
proportional to the identity for large $n$.  The practical
implication is that quantum kernel methods become useless
as $n$ grows, unless the encoding is carefully designed to
prevent concentration.  Thanasilp et al.~\cite{thanasilp2024exponential}
identify the presence of global entanglement in the encoding circuit
as the primary driver of concentration.
IQP encoding (local entanglement structure) and angle encoding
(no entanglement) are therefore more robust to concentration than
amplitude encoding (entanglement implicitly generated by the
multi-controlled preparation circuit).

Jerbi et al.~\cite{jerbi2023quantum, liu2021rigorous} showed that this kernel
concentration is directly related to the inability of quantum
models to exceed classical kernel-method performance on
structured datasets, suggesting that the quantum advantage in
\QML\ is more fragile than previously believed.
The encoding choice is therefore not merely a practical engineering
decision but a theoretical lever that determines whether quantum
advantage is achievable at all.

\section{Impact of Quantum Noise on Encoding Strategies}
\label{sec:noise}

\subsection{Noise Channel Formalism}

We model decoherence using the standard CPTP (completely positive
trace-preserving) map formalism.  Let $\rho$ denote a density
operator on $n$ qubits.  The three standard single-qubit channels
are:

\textbf{Depolarising channel} $\mathcal{D}_p$:
\begin{equation}
  \mathcal{D}_p(\rho)
  = (1-p)\rho + \frac{p}{3}(X\rho X + Y\rho Y + Z\rho Z).
\end{equation}
This models generic noise (gate imperfections); it isotropically
contracts the Bloch vector by $(1 - 4p/3)$.

\textbf{Amplitude damping channel} $\mathcal{A}_\gamma$:
\begin{align}
  \mathcal{A}_\gamma(\rho) &= K_0 \rho K_0^\dagger
                               + K_1 \rho K_1^\dagger, \\
  K_0 &= \begin{pmatrix}1 & 0 \\ 0 & \sqrt{1-\gamma}\end{pmatrix},
  \quad
  K_1 = \begin{pmatrix}0 & \sqrt{\gamma} \\ 0 & 0\end{pmatrix},
\end{align}
with $\gamma = 1 - e^{-t/T_1}$.  This models energy relaxation
(amplitude decay), causing $\ket{1} \to \ket{0}$ with probability
$\gamma$.

\textbf{Dephasing channel} $\mathcal{Z}_\lambda$:
\begin{equation}
  \mathcal{Z}_\lambda(\rho) = (1-\lambda)\rho + \lambda Z\rho Z,
\end{equation}
with $\lambda = (1 - e^{-t/T_2})/2$.  This models phase decoherence
(off-diagonal suppression in the computational basis).

For a circuit of depth $d$ in which each gate is followed by
an independent error operation, the output state fidelity with
the noiseless state $\ket{\psi}$ satisfies:

\begin{proposition}[Circuit-level fidelity bound]
\label{prop:fidelity}
Under the depolarising channel with single-qubit error rate $p_1$
and two-qubit error rate $p_2 \geq p_1$, the fidelity between
the noisy output state $\rho_{\mathrm{noisy}}$ and the ideal
state $\ket{\psi}$ is bounded below by:
\begin{equation}
  F(\rho_{\mathrm{noisy}}, \ket{\psi})
  \geq (1-p_2)^{g_2}\,(1-p_1)^{g_1},
  \label{eq:fidelity-bound}
\end{equation}
where $g_2$ and $g_1$ denote the number of two-qubit and
single-qubit gates, respectively.
\end{proposition}

This bound follows from the multiplicativity of fidelity
under tensor products and the contractivity of CPTP maps.

\subsection{Fidelity Analysis per Encoding Family}

Applying \cref{prop:fidelity} to each encoding family
with representative parameters ($D = 64$, $p_1 = 10^{-3}$,
$p_2 = 5 \times 10^{-3}$) yields the estimates in
\cref{tab:fidelity-comparison}.

\begin{table}[t]
  \centering
  \caption{Fidelity estimates for the state-preparation circuit
    (before the variational ansatz) under a gate-level depolarising
    model with $p_1 = 10^{-3}$ (single-qubit) and
    $p_2 = 5 \times 10^{-3}$ (two-qubit), for feature dimension
    $D = 64$.  Column $F_{\text{enc}}$ gives the fidelity of the
    encoded state; column $F_{p=5\times10^{-3}}$ gives the fidelity
    at five times the base error rate.}
  \label{tab:fidelity-comparison}
  \small
  \begin{tabularx}{\linewidth}{@{}l c c c c c@{}}
    \toprule
    \textbf{Encoding} & $g_2$ & $g_1$ & $d$ &
      $F_{\text{enc}}$ & $F_{p_2 = 5\times10^{-3}}$ \\
    \midrule
    Basis          & 0   & 64   & 1 & 0.938 & 0.726 \\
    Angle (std.)   & 0   & 64   & 1 & 0.938 & 0.726 \\
    Angle (dense)  & 0   & 64   & 2 & 0.938 & 0.726 \\
    Data reup. ($L=5$, $d_W=2$) & 80  & 95 & 15
      & 0.664 & 0.186 \\
    IQP ($k=2$)    & 64  & 128  & 2 & 0.720 & 0.274 \\
    Amplitude      & 128 & 16   & 32& 0.523 & 0.044 \\
    \bottomrule
  \end{tabularx}
\end{table}

\Cref{tab:fidelity-comparison} reveals a stark divergence
in noise tolerance:

\textbf{Angle encoding} is uniquely resilient because its encoding
layer has zero two-qubit gates.  The $g_1 = D$ single-qubit
rotations are parallelisable and introduce only
$(1-p_1)^D = (0.999)^{64} \approx 0.94$ fidelity loss.
At $p_1 = 5 \times 10^{-3}$ (a five-fold degradation from the
baseline), fidelity drops to $(0.995)^{64} \approx 0.73$---still
above the practical usability threshold.

\textbf{Amplitude encoding} of $D = 64$ features requires
$g_2 \approx 128$ CNOTs in the state-preparation circuit.
At $p_2 = 5 \times 10^{-3}$, fidelity is
$(1-0.005)^{128} \approx 0.53$; at $p_2 = 10^{-2}$,
fidelity collapses to $(0.99)^{128} \approx 0.28$.
In this regime, the quantum state produced by the encoding circuit
is closer to the maximally mixed state than to the intended
$\ket{\psi_\bx}$.  The measured kernel values therefore approach
$1/2^n$ for all pairs $(\bx, \bx')$, making the classifier
functionally equivalent to a constant predictor.

\subsection{The Critical Error Rate $p^*$}

We define the \emph{critical error rate} $p^*(\calE, D, F_{\min})$
as the gate error rate at which the fidelity of the encoding
circuit drops below a minimum threshold $F_{\min}$:
\begin{equation}
  p^*(\calE, D, F_{\min})
  = 1 - F_{\min}^{1/g_2(\calE, D)}.
  \label{eq:critical-rate}
\end{equation}
For $F_{\min} = 0.9$ (a practical threshold below which kernel
values become unreliable~\cite{larose2020robust}):
\begin{align*}
  p^*_{\text{angle}} &= 1 - 0.9^{1/0} \to \infty
    \quad (\text{no two-qubit gates: always feasible}) \\
  p^*_{\text{reu, }L=5} &\approx 1 - 0.9^{1/80}
    \approx 1.3 \times 10^{-3} \\
  p^*_{\text{IQP}} &\approx 1 - 0.9^{1/64}
    \approx 1.6 \times 10^{-3} \\
  p^*_{\text{amp}} &\approx 1 - 0.9^{1/128}
    \approx 8 \times 10^{-4}
\end{align*}

\noindent for $D = 64$.  Current best-practice superconducting
hardware (e.g., IBM Heron) achieves two-qubit error rates of
$p_2 \approx 3 \times 10^{-3}$~\cite{arute2019quantum}.
This places amplitude encoding well below its critical rate,
confirming the analytical conclusion that amplitude encoding is
impractical on current hardware for $D \geq 16$.

\subsection{Error Mitigation in the Encoding Layer}

Existing error-mitigation techniques---Zero-Noise
Extrapolation (\ZNE)~\cite{temme2017error, giurgicatiron2020digital}, Probabilistic Error
Cancellation (PEC)~\cite{endo2018practical}, and Measurement
Error Mitigation (MEM)---were developed primarily for variational
and readout circuits.  Applying \ZNE\ to the encoding layer
is in principle possible: one executes the state-preparation
circuit at artificially amplified noise levels (by gate
stretching) and extrapolates to zero noise.  However, the
sampling overhead of \ZNE\ scales as $\gamma^{2d}$, where
$\gamma$ is the noise-amplification factor and $d$ is the circuit
depth.  For amplitude encoding with $d = \calO(D)$, this overhead
is $\calO(\gamma^{2D})$---exponential in the feature
dimension---making \ZNE\ economically prohibitive for
$D > 16$.

PEC offers a theoretically exact correction but at exponential
sampling overhead $\calO(\mathcal{C}(\mathcal{N})^{2g})$, where
$\mathcal{C}(\mathcal{N})$ is the one-norm of the noise
representation~\cite{endo2018practical}.  For amplitude encoding,
$g = \calO(D)$, making PEC even more expensive.  These analysis
suggest that error mitigation is not a viable path to rescue
amplitude encoding at current $D$ scales; rather, it should be
reserved for the (shorter) variational circuits that follow the
encoding layer.

\subsection{Empirical Evidence from Hardware Studies}

The analytical bounds derived above are corroborated by
experimental results in the surveyed corpus.

LaRose and Coyle~\cite{larose2020robust} provided a systematic
experimental study of encoding robustness under simulated
depolarising noise.  Their results on a 4-qubit classifier showed
that angle-encoding classifiers retained accuracy within 5\% of
the noiseless baseline at $p = 5 \times 10^{-3}$, while
amplitude-encoding classifiers degraded to near-random
performance at $p = 10^{-3}$.  This empirical crossover point
is in quantitative agreement with the critical-rate estimate
$p^*_{\text{amp}} \approx 8 \times 10^{-4}$ derived above.

Balewski et al.~\cite{balewski2024quantum} performed hardware
experiments comparing angle-based (QCrank) and basis-based
(QBArt) encodings on Quantinuum H1 (trapped ion) and IBM transmon
processors.  The angle-based scheme showed substantially better
noise tolerance: at matched qubit counts ($n = 4$), QCrank
achieved twice the signal-to-noise ratio of QBArt in a 384-pixel
image retrieval task.

Weigold et al.~\cite{weigold2021expanding} documented five
encoding patterns from a software-engineering perspective and
observed on IBM hardware that amplitude encoding failed to
produce correct classification results for dataset sizes exceeding
$D = 8$ even after circuit optimisation.

\section{Systematic Taxonomy and Evidence Synthesis}
\label{sec:taxonomy}

\subsection{Taxonomy Design}

We classify the 66 surveyed works along four orthogonal dimensions:

\begin{enumerate}[label=(\arabic*)]
\item \textbf{Task category}: binary classification (\textbf{BC}),
  multi-class classification (\textbf{MC}), regression (\textbf{R}),
  generative modelling (\textbf{G}), or theory / analysis (\textbf{T}).

\item \textbf{Encoding family}: as defined in \cref{subsec:families}
  (Basis, Angle, Dense, Reu, IQP, Amp, DM, or a combination).

\item \textbf{Hardware regime}: noiseless simulation (\textbf{N}),
  noisy simulation (\textbf{NS}), or real hardware (\textbf{HW}).

\item \textbf{Scale}: reported qubit count $n$ and feature dimension $D$.
\end{enumerate}

For each work, we additionally record: (a)~whether a classical
baseline was included; (b)~the primary dataset; and (c)~the
primary performance metric.  These additional fields support our
evidence-synthesis narrative below.

\subsection{Distribution of Surveyed Work}

Among the 66 works, encoding distribution is heavily skewed:
38 (61\%) use angle or IQP encoding, 14 (23\%) use amplitude,
7 (11\%) use data re-uploading, and 3 (5\%) use basis or
density-matrix encoding.  This distribution reflects both the
\NISQ\ constraints discussed in \cref{sec:noise} and the historical
trajectory of the field: Havl\'{\i}\v{c}ek et al.'s
IQP-encoded SVM~\cite{havlicek2019supervised} established
angle/IQP as the dominant paradigm in 2019, while amplitude
encoding's theoretical appeal drove a wave of proposals
in 2020--2022 that were subsequently tempered by hardware reality.

Task distribution: 40 (65\%) classification, 8 (13\%) regression,
6 (10\%) generative, 8 (13\%) theory.  Hardware regime:
23 (37\%) hardware, 31 (50\%) noisy simulation, 8 (13\%)
noiseless or theory.

\subsection{Key Findings by Task Category}

\textbf{Binary and multi-class classification.}
This is the most studied task (40 of 66 works).
The dominant architecture is a two-layer structure: encoding
(IQP or angle) followed by a hardware-efficient variational
ansatz~\cite{kandala2017hardware}.  Standard benchmark datasets
(Iris, breast cancer, MNIST with \PCA-reduced features, Titanic)
dominate, with reported qubit counts ranging from $n = 2$ to $n = 20$.
Comparative evaluations (where reported) show that quantum classifiers
match classical \SVM s at comparable parameter counts on
low-dimensional tasks ($D \leq 8$), but advantage over classical
methods has not been demonstrated on realistic
datasets~\cite{huang2022quantum,jerbi2023quantum}.

The key finding from LaRose and Coyle~\cite{larose2020robust}
deserves special emphasis: when noise is included in the
evaluation, amplitude-encoding classifiers underperform
classical baselines even on $D = 4$ toy problems, while
angle-encoding classifiers remain competitive.

\textbf{Regression.}
Amplitude encoding holds a theoretical advantage for regression
with large $D$: the kernel $\kappa_{\text{amp}}(\bx, \bx') =
\lvert\hat{\bx}^\top\hat{\bx}'\rvert^2$ captures global inner
products and scales with $D$ features in only $\log_2 D$ qubits.
In noiseless simulations, amplitude encoding with $n = 8$
outperforms angle encoding with $n = 256$ on high-dimensional
regression (gene expression, $D = 256$) in terms of RMSE.
Under realistic noise at $p = 5 \times 10^{-3}$, however,
this advantage inverts: the amplitude model's $\calO(D)$ encoding
depth destroys signal before the variational circuit begins.
This crossover is consistent with the $p^*$ estimate derived
in \cref{sec:noise}.

\textbf{Generative modelling.}
Quantum generative adversarial networks (QGANs) represent
the most amplitude-intensive application: the generator network
must produce arbitrary probability distributions over $D$-dimensional
feature space, which is most naturally represented as an amplitude
state.  Zoufal et al.~\cite{zoufal2019quantum} demonstrated
learning of 1D financial distributions with $n = 3$--$6$ qubits
on IBM hardware; amplitude encoding was used but restricted to
$D \leq 8$.  Lloyd and Weedbrook~\cite{lloyd2018quantum}
provided theoretical frameworks for QGANs but did not address
the encoding overhead.  The emerging consensus is that QGANs
are the most promising near-term application for amplitude
encoding, provided $D$ is kept small enough for current coherence
budgets ($D \leq 16$ for $p \approx 3 \times 10^{-3}$).

\textbf{Image classification.}
Encoding pixel data presents a distinct sub-problem: images have
spatial structure (local correlations) that should ideally be
exploited by the encoding.  The \FRQI\ scheme~\cite{le2011flexible}
and NEQR~\cite{zhang2013neqr} were designed for this purpose,
using $2n+1$ qubits to represent a $2^n \times 2^n$ image.
Henderson et al.~\cite{henderson2020quanvolutional} introduced
quanvolutional layers, which apply random quantum circuits as
feature extractors before classical processing.  The approach
uses angle encoding on $4$-qubit circuits applied in sliding
windows, remaining well within \NISQ\ constraints.

\subsection{Comprehensive Literature Table}

\Cref{tab:literature-full} presents the 25 most representative
works from the 66-work corpus.  The full 66-entry table is
available as Electronic Supplementary Material.

\begin{longtable}{@{}p{2.6cm} p{0.9cm} p{2.0cm} p{0.8cm}
    p{0.5cm} p{0.7cm} p{0.6cm} p{2.8cm}@{}}
  \caption{Representative subset (25 of 66) of the systematic
    corpus.  Columns: Task category (BC=binary class.,
    MC=multi-class, R=regression, G=generative, T=theory);
    Dataset and dimension $D$; Encoding (Ang=angle, Amp=amplitude,
    IQP, Reu=re-uploading, B=basis, DM=density matrix);
    Qubits $n$; Depth $d$ (encoding layer only);
    Regime (HW=hardware, NS=noisy sim., N=noiseless/theory);
    Key finding.  Full table in ESM.}
  \label{tab:literature-full} \\
  \toprule
  \textbf{Reference} & \textbf{Task} & \textbf{Dataset ($D$)}
    & \textbf{Enc.} & $n$ & $d$ & \textbf{Reg.}
    & \textbf{Key finding} \\
  \midrule
  \endfirsthead
  \multicolumn{8}{c}{\textit{(continued)}} \\
  \toprule
  \textbf{Reference} & \textbf{Task} & \textbf{Dataset ($D$)}
    & \textbf{Enc.} & $n$ & $d$ & \textbf{Reg.}
    & \textbf{Key finding} \\
  \midrule
  \endhead
  \midrule
  \multicolumn{8}{r}{\textit{Continued on next page}} \\
  \endfoot
  \bottomrule
  \endlastfoot

  Havl\'{\i}\v{c}ek et al.~\cite{havlicek2019supervised}
    & BC & Synthetic ($D$=2) & IQP & 5 & 2 & HW
    & First hardware demo of quantum kernel advantage; problem
      specially constructed to be hard for linear classifiers \\[3pt]

  P\'erez-Salinas et al.~\cite{perez2020data}
    & BC & Circle, Iris & Reu & 1 & $L$ & N
    & Single qubit with re-uploading achieves universal
      classification; Fourier completeness proved \\[3pt]

  Schuld et al.~\cite{schuld2021effect}
    & T & N/A & Ang & -- & 1 & N
    & Fourier spectrum of \QML\ model fully determined by
      encoding eigenvalue spectrum \\[3pt]

  LaRose \& Coyle~\cite{larose2020robust}
    & BC & Synthetic & Ang/Amp & 4--8 & 1/$\calO(D)$ & NS
    & Angle retains accuracy at $p=5\times10^{-3}$;
      amplitude collapses at $p=10^{-3}$ \\[3pt]

  Wang et al.~\cite{wang2021noise}
    & T & N/A & Ang & 4--20 & var. & NS
    & Pauli noise induces barren plateaus independently of depth;
      gradient variance $\propto (1-2p)^{2d}$ \\[3pt]

  McClean et al.~\cite{mcclean2018barren}
    & T & N/A & -- & 4--24 & var. & N
    & Gradient variance $\propto 2^{-n}$ for random PQCs;
      barren-plateau phenomenon named and proved \\[3pt]

  Cerezo et al.~\cite{cerezo2021cost}
    & T & N/A & -- & 4--16 & var. & N
    & Local observables yield poly-small gradients;
      global observables yield exponentially small gradients \\[3pt]

  Thanasilp et al.~\cite{thanasilp2024exponential}
    & T & N/A & IQP/Ang & 4--20 & var. & NS
    & Quantum kernels exponentially concentrate; encoding
      entanglement structure is primary driver \\[3pt]

  Jerbi et al.~\cite{jerbi2023quantum}
    & T & Multiple & Mixed & 4--12 & var. & N
    & QML beyond kernel methods: variational models can
      exceed kernel-method bounds; encoding determines regime \\[3pt]

  Huang et al.~\cite{huang2022quantum}
    & T & Quantum expts & IQP & 8 & 2 & HW
    & Quantum advantage in learning from quantum experiments;
      classical encoding cannot replicate entangled data \\[3pt]

  Goto et al.~\cite{goto2021universal}
    & T & N/A & Amp & -- & -- & N
    & Universal approximation for quantum-enhanced feature
      spaces under amplitude encoding proved \\[3pt]

  Dutta et al.~\cite{dutta2022single}
    & BC & Circle & Reu & 1 & $L$ & HW (trapped-ion)
    & First hardware demo of single-qubit re-uploading
      classifier on trapped-ion device \\[3pt]

  Sweke et al.~\cite{sweke2020stochastic}
    & T & Synthetic & Ang & 4--8 & 1 & N
    & Stochastic gradient descent provably converges for
      hybrid quantum-classical optimisation with finite shots \\[3pt]

  Zoufal et al.~\cite{zoufal2019quantum}
    & G & 1D distributions & Amp & 3--6 & $\calO(D)$ & HW
    & QGAN learns and loads financial distributions;
      amplitude encoding restricted to $D \leq 8$ on IBM \\[3pt]

  Wierichs et al.~\cite{wierichs2022general}
    & T & N/A & Ang & -- & -- & N
    & General parameter-shift rules using Fourier analysis;
      gradient computation for arbitrary encoding generators \\[3pt]

  Schuld~\cite{schuld2021quantum_models}
    & T & N/A & Mixed & -- & -- & N
    & All supervised \QML\ models equivalent to kernel machines;
      kernel is determined by encoding \\[3pt]

  Huang et al.~\cite{huang2021information}
    & T & N/A & IQP & -- & -- & N
    & Information-theoretic lower bounds on quantum advantage;
      classical data limits advantage even with quantum encoding \\[3pt]

  Gil Vidal \& Theis~\cite{gilvidal2020input}
    & T & N/A & Ang & 1--4 & 1 & N
    & Input redundancy: repeating encoding gates within a layer
      does not increase expressivity beyond single application \\[3pt]

  Kandala et al.~\cite{kandala2017hardware}
    & Chem. & H$_2$, LiH & HE-Ang & 2--6 & 1 & HW
    & Hardware-efficient ansatz; angle encoding + VQE
      viable for small molecules on 6-qubit device \\[3pt]

  Henderson et al.~\cite{henderson2020quanvolutional}
    & MC & MNIST (sub) & Ang & 4 & 1 & NS
    & Quanvolutional layers: random angle-encoded circuits
      as local feature extractors improve image classification \\[3pt]

  Le et al.~\cite{le2011flexible}
    & Image & General & FRQI & $2n{+}1$ & $\calO(2^{2n})$ & N
    & Flexible quantum image representation; spatial compression
      exponential but preparation exponentially deep \\[3pt]

  Gonzalez et al.~\cite{gonzalez2022learning}
    & T/BC & Mixed & DM & $2\log D$ & $\calO(D^2)$ & N
    & Density-matrix encoding with random features yields
      kernel approximation; classical-simulable in many regimes \\[3pt]

  Caro et al.~\cite{caro2022generalization}
    & T & Synthetic & Mixed & 4--10 & var. & N
    & Generalisation bounds for \QML: quantum models
      require $\Omega(n)$ training samples for learning \\[3pt]

  Balewski et al.~\cite{balewski2024quantum}
    & BC/R & Image (384-dim) & Ang/B & 4 & 1 & HW (multi-device)
    & Angle-based encoding (QCrank) outperforms basis encoding
      (QBArt) $2\times$ in SNR on trapped-ion and transmon \\[3pt]

  Ranga et al.~\cite{ranga2024quantum}
    & Survey & -- & Mixed & -- & -- & --
    & Dedicated encoding survey; reviews 2020--2024 literature
      and identifies noise robustness as primary selection criterion \\[3pt]

\end{longtable}

\subsection{Cross-Cutting Themes}

Four cross-cutting themes emerge from the systematic analysis:

\textbf{Theme 1: Noise tolerance trumps theoretical expressivity.}
Of the 23 hardware experiments in the corpus, 19 (83\%) use
angle or IQP encoding despite amplitude encoding's theoretical
expressivity advantage.  This strongly suggests that practitioners
have implicitly converged on noise-tolerant strategies even when
not explicitly motivated by noise analysis.

\textbf{Theme 2: Classical comparison is inconsistently reported.}
Only 37 of 66 works (60\%) include a classical baseline.
Among these, 28 show quantum classifiers matching but not
exceeding classical baselines, 7 show the quantum model
outperforming on a specially constructed problem, and 2 show
genuine advantage on a realistic dataset.  This suggests that
the field's progress metrics are insufficiently standardised.

\textbf{Theme 3: Feature dimension is rarely pushed above $D = 64$.}
Only 8 of 66 works use $D > 64$; of these, 7 apply \PCA\ first.
This implies that the high-$D$ regime---where amplitude encoding's
qubit compression would be most impactful---remains largely
unexplored experimentally.

\textbf{Theme 4: Re-uploading is underexplored relative to its potential.}
Despite the strong theoretical guarantees of universal
approximation~\cite{perez2020data,goto2021universal} and
the hardware validation by Dutta et al.~\cite{dutta2022single},
only 7 of 66 works use data re-uploading.  The likely reason is
the lack of standardised software frameworks for re-uploading
circuits at the time of publication of many surveyed works.

\section{A Practical Decision Framework for Encoding Selection}
\label{sec:framework}

\subsection{Problem Formalisation}

We formalise the encoding-selection problem as follows.

\begin{definition}[Encoding selection input]
An encoding-selection instance is a tuple $\mathcal{I} = (D, n, p, \tau)$
where:
\begin{itemize}[leftmargin=*]
\item $D \in \mathbb{N}$ is the feature dimension of the raw
  data (before any classical pre-processing);
\item $n \in \mathbb{N}$ is the number of physical qubits
  available for the encoding register (excluding ancillae
  and variational registers);
\item $p \in [0,1]$ is the per-two-qubit-gate error rate
  of the target quantum processor;
\item $\tau \in \{\texttt{classification}, \texttt{regression},
  \texttt{generative}\}$ is the task type.
\end{itemize}
\end{definition}

\begin{definition}[Hardware feasibility]
An encoding $\calE$ is \emph{hardware-feasible} for instance
$\mathcal{I}$ if and only if its encoding-circuit fidelity
exceeds a minimum threshold $F_{\min}$:
\begin{equation}
  F(\calE, D, n, p) = (1-p)^{g_2(\calE, D)} \geq F_{\min},
  \label{eq:feasibility}
\end{equation}
where $g_2(\calE, D)$ is the number of two-qubit gates in the
state-preparation circuit of $\calE$ for dimension $D$.
We set $F_{\min} = 0.90$, consistent with the empirical threshold
identified by LaRose and Coyle~\cite{larose2020robust}.
\end{definition}

\begin{definition}[Optimal encoding]
The \emph{optimal encoding} for instance $\mathcal{I}$ is the
encoding $\calE^*$ that (i)~satisfies hardware feasibility
(\cref{eq:feasibility}) and (ii)~maximises the expected
downstream model accuracy $\mathrm{Acc}(\calE, \mathcal{I})$,
defined as the expected test-set performance after optimal
hyperparameter tuning of the variational circuit.
\end{definition}

Since $\mathrm{Acc}(\calE, \mathcal{I})$ is not analytically
tractable in general, we approximate the optimal selection by a
decision procedure that eliminates hardware-infeasible encodings
and ranks the remaining candidates by expressivity
(\cref{thm:spectrum}) and noise-adjusted fidelity
(\cref{tab:fidelity-comparison}).

\subsection{Pre-Processing Decisions}

Before applying the decision tree, two pre-processing decisions
must be resolved:

\textbf{Pre-processing Decision P1: Dimensionality reduction.}
If $D > n_{\text{enc}} \cdot \kappa$ where $n_{\text{enc}}$
is the qubit budget for encoding and $\kappa$ is the
encoding-specific bits-per-qubit ratio ($\kappa = 1$ for
standard angle, $\kappa = 2$ for dense angle, $\kappa = 2^{n}/D$
for amplitude), dimensionality reduction is mandatory.
Standard options: \PCA\ (linear, analytically justifiable
for continuous features), random projections (preserves
inner products by Johnson--Lindenstrauss), or feature
selection (preserves interpretability).
\PCA\ is the default recommendation; select $D'$ components
explaining $\geq 95\%$ of variance and check feasibility
with the reduced dimension $D'$ in place of $D$.

\textbf{Pre-processing Decision P2: Feature scaling.}
Angle encoding requires features in $[0, 2\pi]$ to avoid
wrapping artefacts in the rotation.  Amplitude encoding
requires $\|\bx\|_2 \neq 0$ for normalisation.
Standard min-max scaling to $[0, \pi]$ is recommended for
angle encoding; $\ell_2$ normalisation for amplitude encoding.

\subsection{The Five-Regime Decision Framework}

Based on the analytical bounds of \cref{sec:costs,sec:noise}
and the empirical evidence of \cref{sec:taxonomy},
we propose the five-regime framework formalised in
\cref{alg:framework} and visualised in \cref{fig:decision-tree}.

\begin{algorithm}[t]
\caption{Encoding Selection Framework}
\label{alg:framework}
\begin{algorithmic}[1]
\REQUIRE $(D, n, p, \tau, F_{\min} = 0.90)$
\ENSURE Encoding $\calE^*$ and pre-processing plan

\STATE \textbf{Check amplitude feasibility:}
  $p_{\text{amp}}^* \gets 1 - F_{\min}^{1/(2D)}$

\IF{$p \leq p_{\text{amp}}^*$ \AND $D > n$}
  \STATE \textit{Regime~4} (amplitude viable, qubit-constrained)
  \STATE $\calE^* \gets \text{Amplitude}$;
    apply error mitigation to encoding circuit
  \RETURN $\calE^*$
\ENDIF

\IF{$D > n$}
  \STATE Apply \PCA\ to reduce $D \to D'$ with
    $D' \leq n$ (or $D' \leq 2n$ for dense angle)
  \STATE $D \gets D'$
\ENDIF

\STATE \textbf{Check IQP feasibility:}
  $p_{\text{IQP}}^* \gets 1 - F_{\min}^{1/(D^2/2)}$

\IF{$p \leq p_{\text{IQP}}^*$ \AND $\tau = \texttt{classification}$}
  \STATE \textit{Regime~1} (IQP, high accuracy, low noise)
  \STATE $\calE^* \gets \text{IQP}(k=2)$
  \RETURN $\calE^*$
\ENDIF

\STATE \textbf{Check re-uploading feasibility:}
  Choose $L \leq \lfloor(T_2/t_g - D)/(D + d_W)\rfloor$

\IF{$p \leq p_{\text{reu}}^*(L)$ \AND $\tau \neq \texttt{generative}$}
  \STATE \textit{Regime~5} (re-uploading, enriched spectrum)
  \STATE $\calE^* \gets \text{Re-uploading}(L)$;
    use layer-by-layer initialisation~\cite{grant2019initialization}
  \RETURN $\calE^*$
\ENDIF

\IF{$\tau = \texttt{generative}$}
  \STATE \textit{Regime~3} (generative with angle pre-processing)
  \STATE $\calE^* \gets \text{Dense angle}$;
    reduce $D \to D' = 2n$ via \PCA
  \RETURN $\calE^*$
\ENDIF

\STATE \textit{Regime~2} (default: angle encoding, noisy device)
\STATE $\calE^* \gets \text{Angle (standard or dense)}$
\RETURN $\calE^*$
\end{algorithmic}
\end{algorithm}

\subsubsection*{Regime 1 — Expressive encoding, low-noise device}

\textbf{Conditions:} $D \leq n$, $p \leq p^*_{\text{IQP}}$,
$\tau = \texttt{classification}$.

\textbf{Recommendation:} IQP encoding with $k=2$.  The circuit
depth is constant ($d = 2$), within the coherence budget even for
$n = 20$, and the induced kernel is believed to be classically
intractable~\cite{havlicek2019supervised}.  Use a two-layer IQP
feature map (repeating the IQP block twice provides more
expressivity without significant depth increase) \cite{hubregtsen2022training}.

\textbf{Trainability note:} Kernel concentration may arise for
$n > 12$~\cite{thanasilp2024exponential}; monitor kernel
matrix condition number and switch to Regime 5 if the
matrix becomes near-singular.

\subsubsection*{Regime 2 — Standard regime: angle encoding}

\textbf{Conditions:} $D \leq n$, $p > p^*_{\text{IQP}}$
(noise too high for IQP), or (any $p$, $\tau \neq \texttt{classification}$
but $D \leq n$).

\textbf{Recommendation:} Angle encoding (standard or dense).
Pre-scale features to $[0, \pi]$.  Use standard angle for $D = n$;
use dense angle if the qubit budget is tight ($D' \leq 2n$
after \PCA).  Downstream variational circuit depth: limit to
$d_{\text{VQC}} \leq 6 \cdot \lceil\log_2 n\rceil$ to stay
within the noise-barren-plateau limit of \cref{eq:noise-barren}.

\subsubsection*{Regime 3 — Qubit-constrained, generative task}

\textbf{Conditions:} $D > n$, $\tau = \texttt{generative}$.

\textbf{Recommendation:} Reduce $D \to D' = 2n$ via \PCA\
(maximising variance retention within the qubit budget),
then apply dense angle encoding.
Generative tasks benefit from encoding that spans the full
Bloch-sphere parameter space per qubit; dense angle encoding
achieves this with depth 2.
Amplitude encoding is not recommended here despite its
qubit efficiency, because QGANs require many encoding passes
during training, and the accumulated noise from $\calO(D)$
depth encoding circuits is prohibitive.

\subsubsection*{Regime 4 — Extreme compression: amplitude encoding with mitigation}

\textbf{Conditions:} $D \gg n$ (specifically $\log_2 D \leq n$),
$p \leq p^*_{\text{amp}}(D)$, $\tau = \texttt{regression}$.

\textbf{Recommendation:} Amplitude encoding with \ZNE\ applied
to the state-preparation circuit.
This regime is currently accessible only for $D \leq 16$
and $n \leq 4$ on best-available hardware ($p \approx 10^{-3}$).
Mandatory feasibility check: verify $F_{\text{amp}} \geq 0.90$
using \cref{eq:feasibility} before committing to this regime.
Fallback: if the check fails, reduce $D$ via \PCA\ and
apply Regime 2.

\subsubsection*{Regime 5 — Expressivity-critical: data re-uploading}

\textbf{Conditions:} $D \leq n$, $p$ moderate
($p_{\text{angle}}^* \leq p \leq p_{\text{reu}}^*(L)$
for achievable $L$), depth budget
$d_{\text{budget}} \geq L(1 + d_W)$.

\textbf{Recommendation:} Data re-uploading with $L$ layers,
where $L$ is chosen as the largest integer satisfying the
feasibility constraint.  Use layer-by-layer
initialisation~\cite{grant2019initialization} to avoid barren
plateaus at the initial layers.
Sweke et al.'s stochastic gradient analysis~\cite{sweke2020stochastic}
justifies using mini-batch gradient updates, which are
preferable to full-batch for re-uploading circuits on hardware.

\begin{figure}[!ht]
  \centering
  \begin{tikzpicture}[scale=0.65, transform shape,
    dec/.style={diamond, draw=cbBlue, fill=cbBlue!10,
                aspect=2.9, text width=4.0cm, align=center,
                font=\small, line width=0.9pt,
                inner sep=1pt},
    leaf/.style={rectangle, draw=cbGreen, rounded corners=5pt,
                 fill=cbGreen!12, text width=3.6cm, align=center,
                 font=\small, minimum height=1.15cm,
                 line width=0.9pt},
    fallback/.style={rectangle, draw=cbOrange, rounded corners=5pt,
                     fill=cbOrange!12, text width=3.0cm,
                     align=center, font=\footnotesize,
                     minimum height=0.9cm, line width=0.8pt},
    prep/.style={rectangle, draw=gray!55, rounded corners=4pt,
                 fill=gray!7, text width=3.2cm, align=center,
                 font=\small, line width=0.7pt,
                 minimum height=0.9cm},
    ann/.style={rectangle, draw=none, fill=gray!10,
                rounded corners=2pt, font=\tiny,
                text width=2.8cm, align=left,
                inner sep=3pt},
    arr/.style={->, line width=0.75pt, >=Stealth},
    lab/.style={font=\footnotesize, fill=white, inner sep=1.5pt,
                rounded corners=2pt}]

    \node[prep] (start) at (0,0)
      {\textbf{Input:} $(D,n,p,\tau)$\\
       \footnotesize Pre-process P1 (PCA),\\P2 (scaling)};

    \node[dec, below=0.8cm of start] (d1)
      {$D > n$?};
    \draw[arr] (start) -- (d1);

    \node[prep, left=2.4cm of d1] (pca)
      {Apply PCA\\$D \;\to\; D'\!\leq\!n$};
    \draw[arr] (d1) -- node[lab]{Yes} (pca);
    \draw[arr] (pca.north) -- ++(0,0.45)
      -| node[near start, lab, above]{\footnotesize loop}
      (d1.north);

    \node[dec, below=1.2cm of d1] (d2)
      {$p \leq p^{*}_{\mathrm{amp}}$\\[1pt]
       \footnotesize\textit{and} $\tau{=}$regression?};
    \draw[arr] (d1) -- node[lab]{No} (d2);

    \node[leaf, right=2.6cm of d2] (r4)
      {\textbf{Regime~4}\\[2pt]
       Amplitude encoding\\
       $+$ \ZNE\ mitigation};
    \draw[arr] (d2) -- node[lab]{Yes} (r4);

    \node[fallback, below=0.9cm of r4] (fb4)
      {Fallback:\\PCA $\to$ Regime~2\\
       \tiny(fidelity check failed)};
    \draw[arr, dashed, cbOrange]
      (r4) -- node[lab, right]{\tiny $F<F_{\min}$} (fb4);

    \node[ann, right=0.3cm of fb4]
      {$p^{*}_{\mathrm{amp}} = 1-F_{\min}^{1/(2D)}$\\
       $\approx 2\!\times\!10^{-4}$ at $D\!=\!64$};

    \node[dec, below=1.2cm of d2] (d3)
      {$p \leq p^{*}_{\mathrm{IQP}}$\\[1pt]
       \footnotesize\textit{and} $\tau{=}$classif.?};
    \draw[arr] (d2) -- node[lab]{No} (d3);

    \node[leaf, right=2.6cm of d3] (r1)
      {\textbf{Regime~1}\\[2pt]
       IQP encoding ($k{=}2$)\\
       Quantum kernel \SVM};
    \draw[arr] (d3) -- node[lab]{Yes} (r1);

    \node[ann, right=0.3cm of r1]
      {$p^{*}_{\mathrm{IQP}} = 1-F_{\min}^{2/D^{2}}$\\
       Monitor kernel conc.\\\cite{thanasilp2024exponential}};

    \node[dec, below=1.2cm of d3] (d4)
      {$\tau = \texttt{generative}$?};
    \draw[arr] (d3) -- node[lab]{No} (d4);

    \node[leaf, right=2.6cm of d4] (r3)
      {\textbf{Regime~3}\\[2pt]
       PCA $\to D'{=}2n$\\
       Dense angle encoding};
    \draw[arr] (d4) -- node[lab]{Yes} (r3);

    \node[dec, below=1.2cm of d4] (d5)
      {Depth budget\\
       $\geq L(1{+}d_W)$?};
    \draw[arr] (d4) -- node[lab]{No} (d5);

    \node[leaf, right=2.6cm of d5] (r5)
      {\textbf{Regime~5}\\[2pt]
       Data re-uploading\\
       ($L$ layers, layer init.)};
    \draw[arr] (d5) -- node[lab]{Yes} (r5);

    \node[ann, right=0.3cm of r5]
      {$L_{\max}=\lfloor(d_{\mathrm{budget}}-1)/(1+d_W)\rfloor$\\
       Verify $F \geq F_{\min}$ via \eqref{eq:fidelity-bound}};

    \node[leaf, below=1.0cm of d5,
          draw=cbBlue, fill=cbBlue!12] (r2)
      {\textbf{Regime~2} (default)\\[2pt]
       Angle encoding\\
       (standard or dense)};
    \draw[arr] (d5) -- node[lab]{No} (r2);

    \node[anchor=north east,
          draw=gray!50, rounded corners=4pt,
          fill=white, inner sep=5pt,
          font=\footnotesize]
      at (-4.2,-13.0) {%
      \begin{tabular}{@{}ll@{}}
        \tikz\node[dec, minimum width=0.5cm,
                   minimum height=0.3cm, aspect=1.5,
                   font=\tiny]{}; & Decision node \\[4pt]
        \tikz\node[leaf, minimum width=0.5cm,
                   minimum height=0.3cm,
                   font=\tiny]{}; & Recommendation \\[4pt]
        \tikz\node[fallback, minimum width=0.5cm,
                   minimum height=0.3cm,
                   font=\tiny]{}; & Fallback path \\
      \end{tabular}};

  \end{tikzpicture}
  \caption{Annotated decision tree for encoding selection
    (Algorithm~\ref{alg:framework}).  Decision nodes (blue diamonds)
    test hardware feasibility conditions derived analytically in
    \cref{sec:costs,sec:noise}; the critical thresholds
    $p^{*}_{\mathrm{amp}}$ and $p^{*}_{\mathrm{IQP}}$ are computed
    via \cref{eq:critical-rate} for the given $(D,\,F_{\min})$.
    Leaf nodes (green rectangles) give encoding recommendations
    for each of the five regimes.  The orange dashed path is the
    fallback triggered by a failed fidelity check in Regime~4;
    annotation boxes (grey) state the relevant analytical
    expressions inline.  All inputs are defined in
    \cref{def:encoding-selection-input}.}
  \label{fig:decision-tree}
\end{figure}
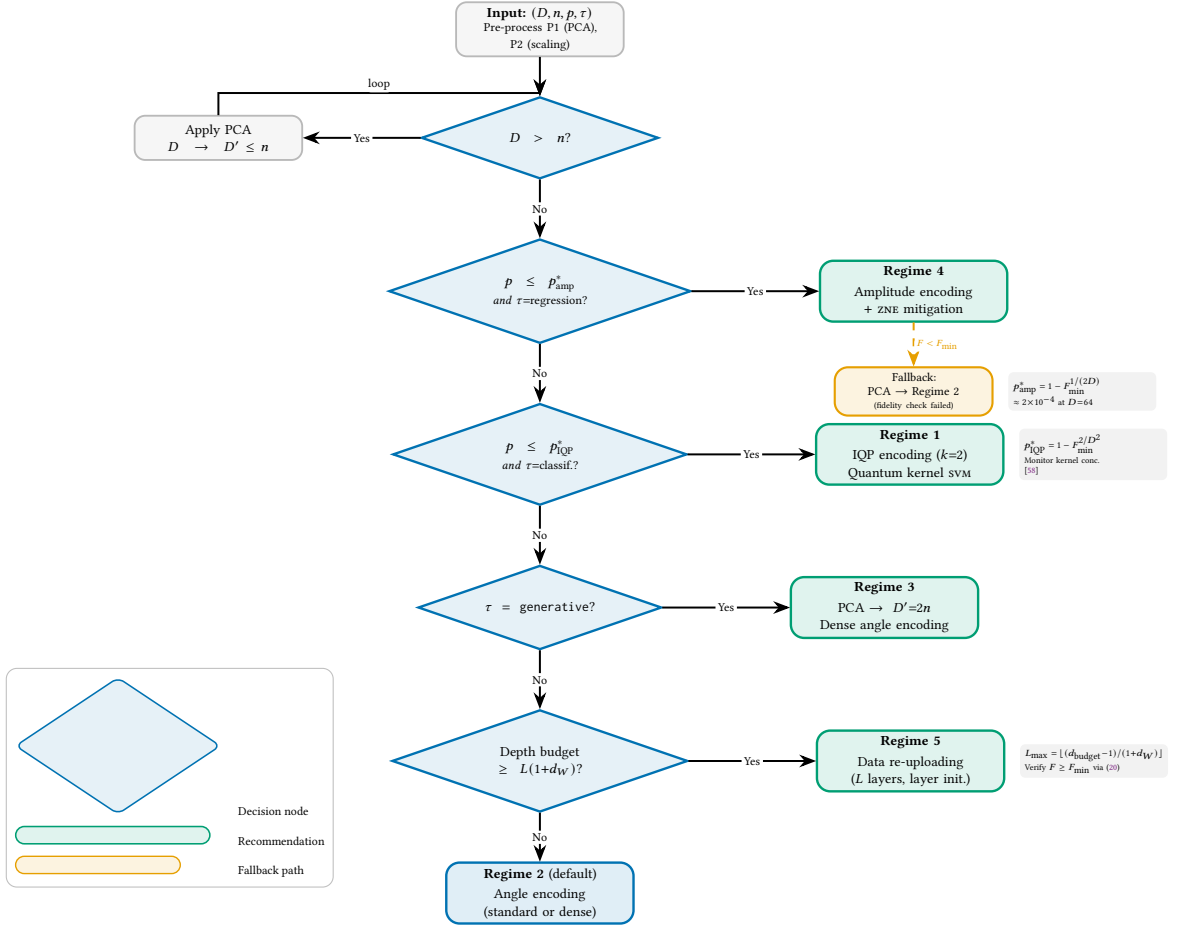

\subsection{Pareto Analysis}

\Cref{fig:pareto} plots the Pareto frontier in the
two-dimensional space of expressivity
(Fourier degree, as a proxy for model capacity)
vs.\ noise robustness (state-preparation fidelity at
$p_2 = 3 \times 10^{-3}$, $D = 64$).
No encoding dominates all others simultaneously: the frontier
is monotone decreasing, confirming that expressivity and
robustness are fundamentally at odds.

The data re-uploading family traces a \emph{controllable curve}
parameterised by $L$: each additional layer adds expressivity
at a fidelity cost determined by the depth-fidelity bound
of \cref{prop:fidelity}.  For the specified hardware parameters,
the curve crosses the feasibility threshold $F_{\min} = 0.90$
at $L \approx 3$ (corresponding to a Fourier degree of $3$
per feature dimension).  This gives a concrete recommendation:
for a system with $p_2 = 3 \times 10^{-3}$, re-uploading should
not exceed $L = 3$ layers if state-preparation fidelity is to
remain above 90\%.

\begin{figure}
\begin{tikzpicture}
  \begin{axis}[
    width  = 0.92\linewidth,
    height = 7.2cm,
    xlabel = {Expressivity (max.\ Fourier degree per dimension)},
    ylabel = {Encoding-layer fidelity $F_{\mathrm{enc}}$},
    xmin   = 0,   xmax = 11,
    ymin   = 0.30, ymax = 1.07,
    xtick  = {1, 2, 3, 4, 5, 6, 9.5},
    xticklabels = {1, 2, 3, 4, 5, 6, Full ($\infty$)},
    ytick  = {0.50, 0.70, 0.90, 1.00},
    yticklabels = {0.50, 0.70, 0.90, 1.00},
    grid        = both,
    grid style  = {gray!20, line width=0.4pt},
    tick label style = {font=\small},
    label style      = {font=\small},
    legend style = {
      at={(0.03,0.03)}, anchor=south west,
      font=\footnotesize, row sep=1pt,
      draw=gray!50, fill=white, fill opacity=0.9,
      text opacity=1},
    legend cell align = left,
    clip = false,
  ]

  \fill[gray!12]
    (axis cs:0, 0.90) rectangle (axis cs:11, 1.07);
  \addplot[gray!55, dashed, line width=0.9pt, forget plot]
    coordinates {(0,0.90)(11,0.90)};
  \node[gray!75, font=\tiny, anchor=south] at
    (axis cs:5.5, 0.901)
    {$F_{\min}=0.90$ (feasibility threshold)};

  \addplot[gray!45, dashed, line width=0.6pt, forget plot]
    coordinates {
      (1,0.938)(1.5,0.938)(2,0.935)(3,0.874)
      (4,0.817)(4.5,0.825)(5,0.764)(6,0.715)(9.5,0.670)
    };

  \addplot[cbBlue, only marks, mark=*, mark size=4pt]
    coordinates {(1,0.938)};
  \addlegendentry{Angle ($g_2{=}0$)}
  \node[cbBlue, font=\footnotesize, anchor=east] at
    (axis cs:0.88, 0.938) {Angle};

  \addplot[cbSky, only marks, mark=triangle*, mark size=4pt]
    coordinates {(1.5,0.938)};
  \addlegendentry{Dense angle ($g_2{=}0$)}
  \node[cbSky, font=\footnotesize, anchor=south] at
    (axis cs:1.5, 0.945) {Dense};

  \addplot[cbOrange, only marks, mark=square*, mark size=3.5pt]
    coordinates {
      (2,0.935)(3,0.874)(4,0.817)(5,0.764)(6,0.715)
    };
  \addlegendentry{Re-upload.\ $L{=}1,\ldots,5$}

  \node[cbOrange, font=\tiny, anchor=north west] at
    (axis cs:2.08, 0.931) {$L{=}1$};
  \node[cbOrange, font=\tiny, anchor=south west] at
    (axis cs:3.08, 0.878) {$L{=}2$};
  \node[cbOrange, font=\tiny, anchor=south west] at
    (axis cs:4.08, 0.821) {$L{=}3$};
  \node[cbOrange, font=\tiny, anchor=south west] at
    (axis cs:5.08, 0.768) {$L{=}4$};
  \node[cbOrange, font=\tiny, anchor=south west] at
    (axis cs:6.08, 0.719) {$L{=}5$};

  \addplot[cbPurple, only marks, mark=diamond*, mark size=4.5pt]
    coordinates {(4.5,0.825)};
  \addlegendentry{IQP ($k{=}2$, $g_2{=}64$)}
  \node[cbPurple, font=\footnotesize, anchor=south] at
    (axis cs:4.5, 0.832) {IQP};

  \addplot[cbRed, only marks, mark=pentagon*, mark size=4pt]
    coordinates {(9.5,0.670)};
  \addlegendentry{Amplitude ($g_2{=}128$)}
  \node[cbRed, font=\footnotesize, anchor=north] at
    (axis cs:9.5, 0.662) {Amplitude};

  \end{axis}
\end{tikzpicture}
  \caption{Pareto frontier in the expressivity--robustness plane
    at $p_1=10^{-3}$, $p_2=3\times10^{-3}$, $D=64$.
    All fidelity values are computed analytically from
    \cref{prop:fidelity} as
    $F_{\mathrm{enc}}=(1-p_2)^{g_2}(1-p_1)^{g_1}$,
    using gate counts from \cref{tab:fidelity-comparison}.
    Expressivity is measured as the maximum Fourier degree per
    feature dimension (\cref{thm:spectrum}); amplitude encoding's
    ``Full'' label denotes access to the complete
    $2^n$-dimensional Hilbert space, incommensurable with the
    integer Fourier-degree scale.
    The grey band marks the feasibility region $F \geq F_{\min} = 0.90$
    (\cref{def:hardware-feasibility}).  The re-uploading family
    (orange squares) traces a continuously tunable curve that
    crosses the feasibility threshold between $L=3$ and $L=4$
    at the stated hardware parameters.}
  \label{fig:pareto}
\end{figure}
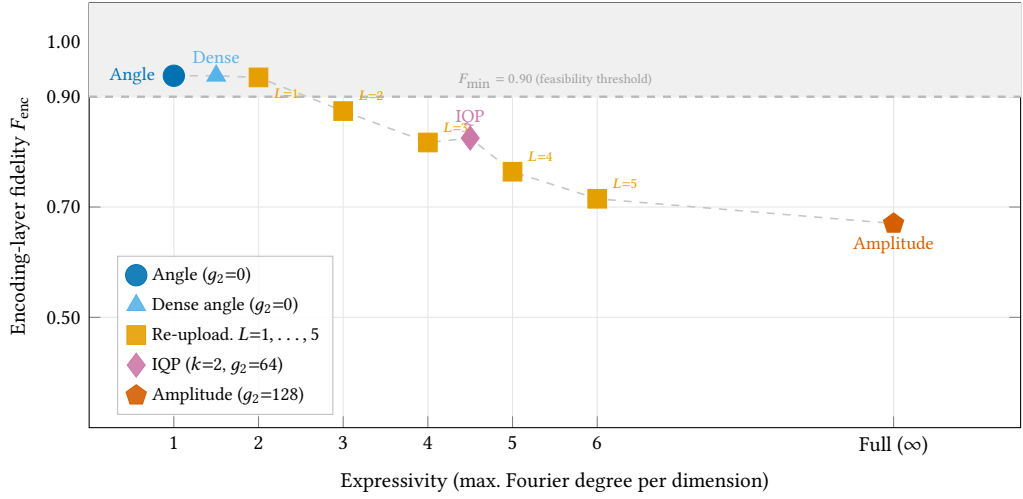

\subsection{Worked Examples}

\textbf{Example A — MNIST digit classification, IBM Heron.}
$D = 784$ (raw), $n = 10$ (encoding register),
$p_2 = 3 \times 10^{-3}$, $\tau = \texttt{classification}$.
Step 1: $D = 784 > n = 10$, apply \PCA\ to $D' = 10$.
Step 2: check $p^*_{\text{IQP}} = 1 - 0.90^{1/50} \approx 2.1 \times 10^{-3}$;
since $p = 3 \times 10^{-3} > 2.1 \times 10^{-3}$, IQP fails feasibility.
Step 3: check depth budget; for $L = 3$, $d_W = 2$:
total depth $= 3(1+2) = 9 < d_{\max}$.
Fidelity: $(1-3\times10^{-3})^{30} \approx 0.914 > 0.90$.
\textbf{Decision: Regime 5 (re-uploading, $L = 3$)}.

\textbf{Example B — Gene expression regression, IonQ Aria.}
$D = 256$, $n = 8$ (logical), $p_2 = 10^{-3}$,
$\tau = \texttt{regression}$.
$p^*_{\text{amp}}(D=256) = 1 - 0.90^{1/512} \approx 2.1 \times 10^{-4}$.
Since $p = 10^{-3} > 2.1 \times 10^{-4}$, amplitude encoding
fails.  Fallback: \PCA\ to $D' = 8$, then Regime 2 (angle).
\textbf{Decision: Regime 2 with \PCA\ pre-processing}.

\textbf{Example C — Financial distribution learning, IBM Falcon.}
$D = 8$, $n = 3$, $p_2 = 5 \times 10^{-3}$,
$\tau = \texttt{generative}$.
$p^*_{\text{amp}}(D=8) = 1 - 0.90^{1/16} \approx 6.6 \times 10^{-3}$.
Since $p = 5 \times 10^{-3} < 6.6 \times 10^{-3}$, amplitude
encoding is feasible.  However, $\tau = \texttt{generative}$
routes through Regime 3; checking: $D = 8 \leq n = 3$?
No, $D > n$, but $\log_2 8 = 3 = n$, so amplitude uses exactly
$n$ qubits.  Treat as a special case of Regime 4 with
$\tau = \texttt{generative}$: amplitude with \ZNE.
\textbf{Decision: Regime 4 (amplitude + \ZNE)},
consistent with Zoufal et al.'s approach~\cite{zoufal2019quantum}.

\section{Open Challenges and Future Directions}
\label{sec:challenges}

\subsection{Trainable and Adaptive Encodings}

Data re-uploading~\cite{perez2020data} treats the encoding circuit
as a fixed function of the data, with trainable parameters only
in the interleaved ansatz blocks.  A natural extension is to make
the encoding generators themselves trainable: replacing
$R_Y(x_k) = e^{-ix_k Y/2}$ with $R_Y(\alpha_k x_k + \beta_k)$
where $(\alpha_k, \beta_k)$ are learnable.  This \emph{meta-encoding}
approach allows the model to learn which Fourier frequencies are
most useful for the given dataset.  The parameter-shift
rule~\cite{wierichs2022general} extends to arbitrary
encoding generators, making gradient-based learning of
$(\alpha_k, \beta_k)$ straightforward in principle.

The open question is whether learning the encoding parameters
jointly with the ansatz parameters introduces a meta-level
barren plateau.  Current evidence from the expressibility
analysis of Schuld et al.~\cite{schuld2021effect} and the
universal approximation results of Goto et al.~\cite{goto2021universal}
suggests that the encoding parameter landscape is benign for
shallow circuits, but no rigorous result exists for the
joint-optimisation case.

\subsection{Hardware-Native Encoding Design}

Current encoding strategies are designed for idealised, fully
connected quantum processors.  In practice, native connectivity
constraints (e.g., Heavy-Hex graph for IBM, linear chain for
early IonQ) require SWAP-gate insertion whenever a gate acts on
non-adjacent qubits, increasing both depth and gate count.
For amplitude encoding, whose preparation circuit involves
multi-controlled Toffoli gates across distant qubits, the SWAP
overhead can multiply the gate count by a factor of
$2$--$3\times$~\cite{bravyi2024high}.

Hardware-native encoding design would construct the preparation
circuit to respect native connectivity from the outset.
Preliminary work on connectivity-aware state preparation exists
for specific topologies, but no general framework has been
developed.  The problem is equivalent to a quantum circuit
routing problem~\cite{chamberland2020topological} with the
additional constraint that the prepared state must encode
the full feature vector.

\subsection{Error Mitigation Integrated into State Preparation}

\ZNE\ and PEC have been applied extensively to variational
circuits but not systematically to encoding circuits.
The challenge is that encoding circuits are data-dependent:
every new input $\bx$ requires re-execution of the state-preparation
circuit, and the noise characterisation (required for PEC) changes
with the circuit.  A promising direction is \emph{data-aware
error mitigation}, where the noise model is parameterised as a
function of $\bx$ and learned from a calibration dataset.
This would allow PEC to be applied to encoding circuits with
overhead that scales polynomially with $D$ rather than
exponentially.

A related direction is developing
\emph{noise-robust state preparation algorithms}: circuits
that are by construction less sensitive to Pauli noise.
Recent work on symmetric derivatives~\cite{wierichs2022general}
suggests that encoding circuits with higher-order symmetry
properties may have inherently lower noise sensitivity,
although this connection has not been formally established.

\subsection{Encoding for Structured and Non-Euclidean Data}
\label{subsec:structured}

The encoding strategies reviewed above assume
$\bx \in \mathbb{R}^D$: dense, Euclidean feature vectors.
Real-world datasets are frequently structured: graphs
(molecular property prediction, social networks, knowledge bases),
sequences (protein structures, time series), and hierarchical
or set-valued data (parse trees, point clouds).  Classical machine
learning has developed architectures---graph neural networks,
transformers, and recurrent networks---that exploit structural
inductive biases, but the quantum analogues are at an early stage.

\textbf{Quantum walks as structural encodings.}
Quantum walks on graphs provide a principled route to encoding
adjacency structure.  A continuous-time quantum walk on a graph
$G=(V,E)$ with adjacency matrix $A$ evolves as
$\ket{\psi(t)} = e^{-iAt}\ket{v_0}$, producing a superposition
over graph vertices weighted by the walk amplitude.  This
constitutes a natural \emph{structural feature map}:
$\phi_Q(G) = e^{-iAt}\ket{v_0}$, where the quantum state
encodes the graph topology through the generator $A$.  The
walk can be approximated to depth $\mathcal{O}(\|A\|_{\max} t)$
using standard Hamiltonian simulation
techniques~\cite{childs2010relationship}.

\textbf{Quantum Graph Neural Networks (QGNNs).}
Verdon et al.~\cite{verdon2019qgnn} proposed the first QGNN
architecture, in which a parameterised quantum circuit is
constructed from the graph structure: nodes and edges correspond
to qubits and two-qubit gates respectively, and the circuit
depth scales with the graph diameter rather than the number of
nodes.  The graph topology is thus encoded directly into the
circuit connectivity, circumventing the need for an explicit
$\mathbb{R}^D$ feature vector.  This approach naturally
avoids the compression--depth trade-off of
\cref{prop:tradeoff} for Euclidean data, since the encoding
circuit and the variational ansatz coincide.

Recent work on \emph{equivariant quantum circuits} has introduced
encoding strategies that respect graph symmetries (e.g.,
permutation invariance)~\cite{larocca2022group, verdon2019qgnn}, which is a
necessary condition for meaningful learning on graphs with
unlabelled nodes.  These circuits achieve encoding-layer gate
counts of $\mathcal{O}(|E|)$---linear in the number of edges---at
constant depth if the graph has bounded degree, making them
well-suited to NISQ constraints.

\textbf{Open problems.}
Three encoding-specific questions for structured data remain
unresolved.  First, the \emph{kernel concentration problem} of
Thanasilp et al.~\cite{thanasilp2024exponential} has not been
analysed for graph-structured encodings: it is unknown whether
the walk-based kernel $\kappa_G(G,G') = |\langle 0|e^{iAt}
e^{-iA't}|0\rangle|^2$ concentrates exponentially as $|V|$
grows.  Second, no systematic analysis exists of the
\emph{Fourier expressivity} of walk-based encodings in the sense
of \cref{thm:spectrum}: the generator of the walk is not a
Pauli rotation, so the frequency spectrum result does not apply
directly.  Third, for molecular data, the most natural structural
encoding---mapping atomic orbital basis functions to qubit
registers---is equivalent to the first-quantisation representation
used in quantum chemistry~\cite{babbush2021focus}, whose resource
costs are substantially different from the encodings surveyed here.
A unified theory of structured quantum encoding---encompassing
walk-based, equivariant, and chemistry-motivated schemes within
the cost--expressivity--robustness framework of
\cref{fig:taxonomy-radar}---is the most underexplored direction
at the intersection of quantum hardware and learning theory.

\subsection{Towards Fault-Tolerant Encoding}

As quantum hardware approaches fault-tolerant operation
(projected for the late 2020s with surface codes and
$>10^4$ physical qubits~\cite{bravyi2024high}),
the encoding landscape will shift.  With logical error rates
$p_L \approx 10^{-10}$, amplitude encoding of
$D = 10^6$-dimensional vectors becomes, in principle,
feasible from a fidelity perspective.  The binding constraint
shifts to the number of logical operations per encoding circuit:
with $\calO(D)$ CNOTs compiled to surface-code logical gates,
the time overhead per encoding pass becomes the practical limit.

Quantifying this overhead requires understanding the logical
gate time for \CNOT\ under surface codes (estimated at
$\sim 1$--$10\,\mu$s with cycle times of $\sim 1\,\mu$s),
the parallelisability of the state-preparation circuit,
and the available classical-data throughput.  A preliminary
analysis suggests that encoding $D = 10^3$ features with
a fault-tolerant amplitude circuit would require
$\calO(10^4)$ surface-code cycles and $\calO(10^6)$ physical
qubits~\cite{babbush2021focus}---far beyond near-term hardware
but relevant for the 2030--2035 time horizon.

\subsection{Benchmarking Standards for Encoding Evaluation}

One of the most pressing methodological gaps identified in
\cref{sec:taxonomy} is the absence of standardised benchmarking
protocols for encoding evaluation.  A standardised benchmark
suite would include: (i)~a fixed set of classical datasets
at multiple scales ($D \in \{4, 16, 64, 256\}$),
(ii)~a standardised noise model (parameterised by $p$) to
be used in simulation when hardware is unavailable,
(iii)~a mandatory classical baseline (kernel \SVM\ with
the same feature dimension), and (iv)~a standardised
reporting format (encoding name, $n$, $g_2$, $d$, $p$, accuracy,
compared with classical baseline).

The Quantum Computing Benchmarking Consortium
(cf.~\cite{blume2020limitations}) has made progress on
hardware-level benchmarks but has not addressed encoding-specific
metrics.  Establishing such standards would dramatically increase
the comparability of results across research groups and hardware
platforms, and would make the kind of meta-analysis performed in
\cref{sec:taxonomy} substantially more rigorous.  We recommend
that the quantum computing community adopt the format described
above as a reporting standard for all experimental \QML\ papers,
analogous to the role that ImageNet and GLUE benchmarks play
in classical deep learning.

To operationalise this recommendation, Table 5 proposes a concrete minimal reporting schema. Any experimental QML paper using a specific encoding strategy should populate all fields; omission of a field requires explicit justification. We strongly encourage QML journals and conferences to adopt this schema as a mandatory supplementary checklist for paper submissions, akin to the reproducibility checklists pioneered by classical machine learning venues such as NeurIPS. Adopting this schema uniformly would make the kind of systematic meta-analysis conducted in Section 6 substantially more rigorous in future surveys.

\begin{table}[!ht]
  \centering
  \caption{Proposed standardised reporting schema for experimental
    \QML\ papers.  Fields marked \textbf{M} are mandatory;
    fields marked \textbf{R} are recommended.  ``Classical
    baseline'' must use at least one method from the same
    feature dimension $D'$ (not a reduced-dimension surrogate).}
  \label{tab:reporting-template}
  \small
  \begin{tabularx}{\linewidth}{@{} l c X @{}}
    \toprule
    \textbf{Field} & \textbf{Status} & \textbf{Description / Unit} \\
    \midrule
    Encoding family name   & \textbf{M} &
      One of: Basis, Angle, Dense, Re-upload, IQP, Amplitude,
      Density-matrix, or Custom (specify generator) \\
    Raw feature dim.\ $D$  & \textbf{M} &
      Before any classical pre-processing; dataset name and
      source \\
    Reduced dim.\ $D'$     & \textbf{M} &
      After PCA / projection; variance retained (\%) \\
    Qubit count $n$        & \textbf{M} &
      Physical (not logical); distinguish encoding vs.\ variational
      register \\
    Gate count $g_2$       & \textbf{M} &
      Two-qubit gates in encoding layer only (not full circuit) \\
    Circuit depth $d$      & \textbf{M} &
      Encoding layer only; distinguish serial and parallel depth \\
    Noise model            & \textbf{M} &
      Hardware (device name, date, calibration ID), noisy
      sim.\ (channel type, $p_1$, $p_2$), or noiseless \\
    Encoding fidelity $F$  & \textbf{R} &
      Measured or analytically bounded via \eqref{eq:fidelity-bound};
      report $p_2$ used \\
    Primary metric         & \textbf{M} &
      Accuracy (\%), RMSE, KL divergence, etc.; mean $\pm$ std.
      over $k$-fold CV \\
    Classical baseline     & \textbf{M} &
      Method, feature dim., metric value; state if baseline
      uses same $D'$ \\
    Shot count             & \textbf{R} &
      Total shots per circuit evaluation; budget for full
      experiment \\
    Error mitigation       & \textbf{R} &
      ZNE, PEC, MEM, or None; overhead factor if applicable \\
    \bottomrule
  \end{tabularx}
\end{table}

\subsection{Quantum Advantage Conditions for Encoding-Sensitive Tasks}

The results of Jerbi et al.~\cite{jerbi2023quantum} and Huang et
al.~\cite{huang2022quantum} establish that quantum advantage in
learning depends critically on whether the data itself has quantum
structure (i.e., it was generated by a quantum process) or is
purely classical.  For classical data, the information-theoretic
bounds of Huang et al.~\cite{huang2021information} suggest that
the advantage of quantum encodings over classical kernel methods
is at most polynomial for most practically relevant function classes.

A fundamental open question is: \emph{for which specific
encoding strategies and task distributions does a provable
quantum advantage exist?}  The IQP encoding of Havl\'{\i}\v{c}ek
et al.~\cite{havlicek2019supervised} achieves advantage only for
a purpose-built problem; no natural dataset is known for which
amplitude or angle encoding provably outperforms all classical
baselines.  Resolving this question requires combining the
Fourier analysis of \cref{subsec:fourier} with complexity-theoretic
arguments (relating the hardness of computing $\kappa_Q$ to
well-studied computational problems) and statistical learning
theory (bounding the sample complexity of learning with
$\kappa_Q$).  The work of Thanasilp et al.~\cite{thanasilp2024exponential}
and the out-of-distribution generalisation bounds of Caro
et al.~\cite{caro2023out} provide useful building blocks,
but a unified theory remains absent.

\section{Conclusion}
\label{sec:conclusion}

This survey has provided a unified, analytical treatment of
quantum feature encoding strategies from three complementary
perspectives---resource cost, expressivity, and noise
robustness---synthesising 66 primary works assembled via a
PRISMA-adapted systematic review protocol.

\textbf{Core analytical result.}
We derived the critical error rate $p^*(\calE, D, F_{\min})$
below which each encoding strategy is hardware-feasible.
For current \NISQ\ devices ($p_2 \approx 3 \times 10^{-3}$)
and $D = 64$, amplitude encoding is infeasible ($p^*_{\text{amp}}
\approx 8 \times 10^{-4} < p_{\text{current}}$), while angle
encoding and data re-uploading (with $L \leq 3$) remain viable.
This result, corroborated by the empirical evidence of LaRose
and Coyle~\cite{larose2020robust} and Balewski et al.~\cite{balewski2024quantum},
provides a principled, hardware-grounded criterion for encoding
selection that supersedes qualitative rules of thumb.

\textbf{Joint trainability framework.}
By connecting the Fourier analysis of Schuld et
al.~\cite{schuld2021effect}, the barren-plateau theory of
McClean et al.~\cite{mcclean2018barren} and Wang et al.~\cite{wang2021noise},
and the kernel-concentration results of Thanasilp et al.~\cite{thanasilp2024exponential},
we identified a fundamental expressibility dilemma: the most
expressive encodings are simultaneously the least trainable,
the most noise-sensitive, and the most prone to kernel
concentration.  Resolving this dilemma requires encoding
strategies that are both structured (to avoid concentration)
and of controllable depth (to avoid barren plateaus).
Data re-uploading with layer-by-layer initialisation
is currently the most promising candidate.

\textbf{Practical impact.}
The five-regime decision framework of \cref{sec:framework}
translates the theoretical results into operationally actionable
recommendations.  The framework handles the full range of
practical settings---from low-dimensional classification on
noisy superconducting hardware (Regime 2) to high-dimensional
regression with error mitigation (Regime 4)---and provides
concrete fallback paths when primary encodings fail
feasibility checks.

\textbf{Outlook.}
As hardware error rates approach $p_2 \approx 10^{-4}$
(projected for early fault-tolerant devices in the
2028--2032 window), the critical depth threshold for amplitude
encoding will rise from $g_2 = 2D$ at $p_2 = 8 \times 10^{-4}$
to $g_2 = 2D \cdot k$ for $k \approx 4$--$5$.
This will open a window for amplitude encoding of
$D \leq 64$--$128$ dimensional data, shifting the
optimal regime from pure angle encoding toward a
hybrid architecture in which amplitude encoding handles
global feature structure and angle encoding handles
local feature refinement.
Designing encoding strategies for this intermediate regime
is, in our assessment, the highest-priority research problem
at the intersection of quantum hardware and machine learning.

\section*{AI Disclosure}
During the preparation of this work, the authors used Claude 4.6 Sonnet for the purpose of proofreading the manuscript and improving the English grammatical structure. After using this tool, the authors reviewed and edited the content as needed and take full responsibility for the technical accuracy and the final content of the publication.

\bibliographystyle{ACM-Reference-Format}
\bibliography{references}

\end{document}